\documentclass[11pt]{article}
\usepackage{mathtools}
\usepackage{amsmath}
\usepackage[shortlabels]{enumitem}
\usepackage[dvipsnames]{xcolor}
\usepackage{framed}
\definecolor{shadecolor}{rgb}{0.9,0.9,0.9}

\usepackage{graphicx,epic,eepic,epsfig,amsmath,latexsym,amssymb,verbatim,color}
 
\usepackage{amsfonts}       % blackboard math symbols
\usepackage{nicefrac}       % compact symbols for 1/2, etc.
\usepackage{mathrsfs}

\usepackage{amsmath}
\usepackage{bbm}

\usepackage{float}
\usepackage{tikz}
\usetikzlibrary{chains}
\usetikzlibrary{fit}
\usetikzlibrary{arrows}
\usetikzlibrary{decorations}

\usepackage{epsfig}
\usetikzlibrary{shapes.symbols,patterns} % for source symbols
\usepackage{pgfplots}

\usepackage[strict]{changepage}
\usepackage{hyperref}
\hypersetup{colorlinks=true,citecolor=blue,linkcolor=blue,filecolor=blue,urlcolor=blue,breaklinks=true}

\usepackage[marginal]{footmisc}
\usepackage{url}
\usepackage{theorem}

\newtheorem{definition}{Definition}
\newtheorem{proposition}{Proposition}
\newtheorem{lemma}[proposition]{Lemma}

\newtheorem{theorem}[proposition]{Theorem}

\def\squareforqed{\hbox{\rlap{$\sqcap$}$\sqcup$}}
\def\qed{\ifmmode\squareforqed\else{\unskip\nobreak\hfil
\penalty50\hskip1em\null\nobreak\hfil\squareforqed
\parfillskip=0pt\finalhyphendemerits=0\endgraf}\fi}
\def\endenv{\ifmmode\;\else{\unskip\nobreak\hfil
\penalty50\hskip1em\null\nobreak\hfil\;
\parfillskip=0pt\finalhyphendemerits=0\endgraf}\fi}
\newenvironment{proof}{\noindent \textbf{{Proof~} }}{\hfill $\blacksquare$}

\newcounter{remark}
\newenvironment{remark}[1][]{\refstepcounter{remark}\par\medskip\noindent%
\textbf{Remark~\theremark #1} }{\medskip}

\newcounter{example}

\mathchardef\ordinarycolon\mathcode`\:
\mathcode`\:=\string"8000
\def\vcentcolon{\mathrel{\mathop\ordinarycolon}}
\begingroup \catcode`\:=\active
  \lowercase{\endgroup
  \let :\vcentcolon
  }

\usepackage{cleveref}
\usepackage{graphicx}
\usepackage{xcolor}

\RequirePackage[framemethod=default]{mdframed}
\newmdenv[skipabove=7pt,
skipbelow=7pt,
backgroundcolor=darkblue!15,
innerleftmargin=5pt,
innerrightmargin=5pt,
innertopmargin=5pt,
leftmargin=0cm,
rightmargin=0cm,
innerbottommargin=5pt,
linewidth=1pt]{tBox}

\newmdenv[skipabove=7pt,
skipbelow=7pt,
backgroundcolor=darkred!15,
innerleftmargin=5pt,
innerrightmargin=5pt,
innertopmargin=5pt,
leftmargin=0cm,
rightmargin=0cm,
innerbottommargin=5pt,
linewidth=1pt]{rBox}

\newmdenv[skipabove=7pt,
skipbelow=7pt,
backgroundcolor=blue2!25,
innerleftmargin=5pt,
innerrightmargin=5pt,
innertopmargin=5pt,
leftmargin=0cm,
rightmargin=0cm,
innerbottommargin=5pt,
linewidth=1pt]{dBox}
\newmdenv[skipabove=7pt,
skipbelow=7pt,
backgroundcolor=darkkblue!15,
innerleftmargin=5pt,
innerrightmargin=5pt,
innertopmargin=5pt,
leftmargin=0cm,
rightmargin=0cm,
innerbottommargin=5pt,
linewidth=1pt]{sBox}
\definecolor{darkblue}{RGB}{0,76,156}
\definecolor{darkkblue}{RGB}{0,0,153}
\definecolor{blue2}{RGB}{102,178,255}
\definecolor{darkred}{RGB}{195,0,0}
\newcommand{\nc}{\newcommand}
\nc{\rnc}{\renewcommand}
\nc{\lbar}[1]{\overline{#1}}
\nc{\bra}[1]{\langle#1|}
\nc{\ket}[1]{|#1\rangle}
\nc{\ketbra}[2]{|#1\rangle\!\langle#2|}
\nc{\braket}[2]{\langle#1|#2\rangle}

\nc{\proj}[1]{| #1\rangle\!\langle #1 |}
\nc{\avg}[1]{\langle#1\rangle}
\nc{\rank}{\operatorname{Rank}}
\nc{\smfrac}[2]{\mbox{$\frac{#1}{#2}$}}
\nc{\tr}{\operatorname{Tr}}
\nc{\ox}{\otimes}
\nc{\dg}{\dagger}
\nc{\dn}{\downarrow}
\nc{\cA}{{\cal A}}
\nc{\cB}{{\cal B}}
\nc{\cC}{{\cal C}}
\nc{\cD}{{\cal D}}
\nc{\cE}{{\cal E}}
\nc{\cF}{{\cal F}}
\nc{\cG}{{\cal G}}
\nc{\cH}{{\cal H}}
\nc{\cI}{{\cal I}}
\nc{\cJ}{{\cal J}}
\nc{\cK}{{\cal K}}
\nc{\cL}{{\cal L}}
\nc{\cM}{{\cal M}}
\nc{\cN}{{\cal N}}
\nc{\cO}{{\cal O}}
\nc{\cP}{{\cal P}}
\nc{\cQ}{{\cal Q}}
\nc{\cR}{{\cal R}}
\nc{\cS}{{\cal S}}
\nc{\cT}{{\cal T}}
\nc{\cU}{{\cal U}}
\nc{\cV}{{\cal V}}
\nc{\cX}{{\cal X}}
\nc{\cY}{{\cal Y}}
\nc{\cZ}{{\cal Z}}
\nc{\cW}{{\cal W}}
\nc{\csupp}{{\operatorname{csupp}}}
\nc{\qsupp}{{\operatorname{qsupp}}}
\nc{\var}{{\operatorname{var}}}
\nc{\rar}{\rightarrow}
\nc{\lrar}{\longrightarrow}
\nc{\polylog}{{\operatorname{polylog}}}
\nc{\wt}{{\operatorname{wt}}}
\nc{\av}[1]{{\left\langle {#1} \right\rangle}}
\nc{\supp}{{\operatorname{supp}}}

\nc{\argmin}{{\operatorname{argmin}}}

\def\x{\xi}

\nc{\RR}{{{\mathbb R}}}
\nc{\CC}{{{\mathbb C}}}
\nc{\FF}{{{\mathbb F}}}
\nc{\NN}{{{\mathbb N}}}
\nc{\ZZ}{{{\mathbb Z}}}
\nc{\PP}{{{\mathbb P}}}
\nc{\QQ}{{{\mathbb Q}}}
\nc{\UU}{{{\mathbb U}}}
\nc{\EE}{{{\mathbb E}}}
\nc{\id}{{\operatorname{id}}}

\nc{\CHSH}{{\operatorname{CHSH}}}

\nc{\be}{\begin{equation}}
\nc{\ee}{{\end{equation}}}
\nc{\bea}{\begin{eqnarray}}
\nc{\eea}{\end{eqnarray}}
\nc{\<}{\langle}
\rnc{\>}{\rangle}
\nc{\rU}{\mbox{U}}

\nc{\ob}[1]{#1}

\nc{\OLOCC}{{\text{1-LOCC}}}
\nc{\SEP}{{\text{SEP}}}
\nc{\NS}{{\text{NS}}}
\nc{\LOCC}{{\text{LOCC}}}
\nc{\PPT}{{\text{PPT}}}
\nc{\EXT}{{\text{EXT}}}
\nc{\Sym}{{\operatorname{Sym}}}

\nc{\ERLO}{{E_{{ R,LO}}}}
\nc{\ERLOCC}{{E_{{R,\text{PPT}}}}}
\nc{\ERPPT}{{E_{{R,\text{PPT}}}}}
\nc{\ERPPTinf}{{E^{\infty}_{{R,\text{PPT}}}}}
\nc{\ER}{E_{\rm R}}
\nc{\ERLOCCinfty}{{E^{\infty}_{{r,LOCC}}}}
\nc{\Aram}{{\operatorname{\sf A}}}
\nc{\ECPPT}{{E_{{C,\text{PPT}}}}}
\nc{\EDPPT}{{E_{{D,\text{PPT}}}}}
\nc{\Freek}{{\text{PPT$_k$}}}
\nc{\Freesec}{{\text{PPT$_2$}}}

\nc{\NB}{{{{\tiny N}}}}
\nc{\LB}{{{LN}}}
\nc{\NPT}{{\text{NPT}}}

\usepackage{tikz}
\usepackage{hyperref}
\hypersetup{colorlinks=true,citecolor=blue,linkcolor=blue,filecolor=blue,urlcolor=blue,breaklinks=true}

\makeatletter
\def\grd@save@target#1{%
  \def\grd@target{#1}}
\def\grd@save@start#1{%
  \def\grd@start{#1}}
\tikzset{
  grid with coordinates/.style={
    to path={%
      \pgfextra{%
        \edef\grd@@target{(\tikztotarget)}%
        \tikz@scan@one@point\grd@save@target\grd@@target\relax
        \edef\grd@@start{(\tikztostart)}%
        \tikz@scan@one@point\grd@save@start\grd@@start\relax
        \draw[minor help lines,magenta] (\tikztostart) grid (\tikztotarget);
        \draw[major help lines] (\tikztostart) grid (\tikztotarget);
        \grd@start
        \pgfmathsetmacro{\grd@xa}{\the\pgf@x/1cm}
        \pgfmathsetmacro{\grd@ya}{\the\pgf@y/1cm}
        \grd@target
        \pgfmathsetmacro{\grd@xb}{\the\pgf@x/1cm}
        \pgfmathsetmacro{\grd@yb}{\the\pgf@y/1cm}
        \pgfmathsetmacro{\grd@xc}{\grd@xa + \pgfkeysvalueof{/tikz/grid with coordinates/major step}}
        \pgfmathsetmacro{\grd@yc}{\grd@ya + \pgfkeysvalueof{/tikz/grid with coordinates/major step}}
        \foreach \x in {\grd@xa,\grd@xc,...,\grd@xb}
        \node[anchor=north] at (\x,\grd@ya) {\pgfmathprintnumber{\x}};
        \foreach \y in {\grd@ya,\grd@yc,...,\grd@yb}
        \node[anchor=east] at (\grd@xa,\y) {\pgfmathprintnumber{\y}};
      }
    }
  },
  minor help lines/.style={
    help lines,
    step=\pgfkeysvalueof{/tikz/grid with coordinates/minor step}
  },
  major help lines/.style={
    help lines,
    line width=\pgfkeysvalueof{/tikz/grid with coordinates/major line width},
    step=\pgfkeysvalueof{/tikz/grid with coordinates/major step}
  },
  grid with coordinates/.cd,
  minor step/.initial=.2,
  major step/.initial=1,
  major line width/.initial=2pt,
}
\makeatother

\usepackage{thmtools}
\usepackage{thm-restate}
\usepackage{etoolbox}
\makeatletter
\def\problem@s{}
\newcounter{problems@cnt}

\newcommand{\allproblems}{\problem@s}
\makeatother

\usepackage{amsmath,amsfonts,amssymb,array,graphicx,subfigure,float,multirow,bm,times,tcolorbox,relsize,booktabs}
\usepackage[utf8]{inputenc}
\usepackage[T1]{fontenc}
\usepackage{qcircuit}
\usepackage{authblk} % package for author list
\usepackage[paperwidth=190.8mm,paperheight=297mm,centering,hmargin=20mm,vmargin=2cm]{geometry}
\usepackage[ruled, vlined, linesnumbered]{algorithm2e}
\usepackage{pgfplots}
\usepackage{diagbox}
\usepackage{utfsym}
\usepackage{fontawesome}
\usepackage{wasysym}
\usepackage{threeparttable}

\pgfplotsset{compat=1.9}

\definecolor{colortwo}{rgb}{0.4,0.77,0.17}
\definecolor{colorthree}{rgb}{0.01,0.51,0.93}
\newcommand*\samethanks[1][\value{footnote}]{\footnotemark[#1]}

\allowdisplaybreaks

\begin{document}
\title{Protocols and Trade-Offs of Quantum State Purification}

\author[1]{Hongshun Yao\thanks{Hongshun Yao and Yu-Ao Chen contributed equally to this work.}}
\author[1]{Yu-Ao Chen\samethanks[1]}
\author[1]{Erdong Huang}
\author[1]{Kaichu Chen}
\author[2,3]{Honghao Fu
\thanks{honghao.fu@concordia.ca}}
\author[1]{Xin Wang \thanks{felixxinwang@hkust-gz.edu.cn}}

\affil[1]{\small Thrust of Artificial Intelligence, Information Hub,\par The Hong Kong University of Science and Technology (Guangzhou), Guangdong 511453, China.}
\affil[2]{\small Concordia Institute for Information Systems Engineering, \par Concordia University, Quebec H3G 1M8, Canada.}
\affil[3]{\small Computer Science and Artificial Intelligence Lab, \par Massachusetts Institute of Technology, Massachusetts 02139, USA.}

\date{\today}
\maketitle

\begin{abstract}
Quantum state purification is crucial in quantum communication and computation, aiming to recover a purified state from multiple copies of an unknown noisy state.
This work introduces a general state purification framework designed to achieve the highest fidelity with a specified probability and characterize the associated trade-offs. For i.i.d. quantum states under depolarizing noise, our framework can replicate the purification protocol proposed by [Barenco et al., SIAM Journal on Computing, 26(5), 1997] and further provide exact formulas for the purification fidelity and probability with explicit trade-offs. We prove the protocols' optimality for two copies of noisy states with any dimension and confirm its optimality for higher numbers of copies and dimensions through numerical analysis.
Our methodological approach paves the way for proving the protocol's optimality in more general scenarios and leads to optimal protocols for other noise models. Furthermore, we present a systematic implementation method via block encoding and parameterized quantum circuits, providing explicit circuits for purifying three-copy and four-copy states under depolarizing noise. Finally, we estimate the sample complexity and generalize the protocol to a recursive form, demonstrating its practicality for quantum computers with limited memory.
\end{abstract}

\tableofcontents
\newpage

%%%%%%%%%%%%%%%%%%%%%%%%%%%%%%%%%%%%%%%%%%%%%%%%%%%%%%%%%%%%%%%%%%%%%%%%%%%%
%%%%%%%%%%%%%%%%%%%%%%%%%%%%%%%%%%%%%%%%%%%%%%%%%%%%%%%%%%%%%%%%%%%%%%%%%%%%
\section{Introduction}\label{sec:introduction}
%%%%%%%%%%%%%%%%%%%%%%%%%%%%%%%%%%%%%%%%%%%%%%%%%%%%%%%%%%%%%%%%%%%%%%%%%%%%

The pursuit of quantum technologies with superior performance is central to the advancement of quantum information processing. However, the intrinsic susceptibility of quantum devices to environmental disturbances significantly hinders their operational integrity, thereby constraining the demonstration of quantum advantages~\cite{arute2019quantum,preskill2018quantum} in various domains, including quantum communication and computation. To address this challenge, a suite of strategies has been devised for the preservation of quantum states against quantum noises. Quantum error correction~\cite{shor1995scheme, steane1996error} is one such technique that addresses this issue by leveraging an expanded Hilbert space to encode quantum information.
On a parallel front, quantum state purification~\cite{bennett1996purification,cirac1999optimal,keyl2001rate} emerges as a vital technique, particularly within the context of quantum information processing and near-term quantum computing. 

Quantum state purification aims to generate the quantum state with enhanced purity from multiple copies of an unknown noisy state. Theoretical studies have extensively explored purification protocols, particularly in the context of noise arising from the depolarizing channel~\cite{cirac1999optimal,keyl2001rate}, which is a common type of noise encountered in quantum systems.
One of the most celebrated purification protocols is the Cirac-Ekert-Macchiavello (CEM) protocol~~\cite{cirac1999optimal}. This protocol has been shown to achieve the optimal average fidelity between the output state and the ideal pure state, meaning that the fidelity obtained after applying the CEM protocol cannot be further improved by any other purification protocols without consuming more noisy copies. Inspired by the theoretical success of the CEM protocol, several experimental works have been proposed to demonstrate its practical implementation in various experimental settings~\cite{ricci2004experimental,hou2014experimental}. Specifically, the work~\cite{ricci2004experimental} first implemented the CEM protocol in a linear-optical system for the case of two input states, i.e., two depolarized states are consumed in this protocol. Another work~\cite{hou2014experimental} demonstrated the CEM protocol in the nuclear magnetic resonance system. 

General purification protocols, such as the CEM protocol, focus on the optimal average fidelity $F=\sum_{j}p_jf_j$ where $j$ ranges over all the measurement outcomes, $p_j$ is the probability of outcome $j$
and $f_j$ is the achieved fidelity when measuring $j$. While in a probabilistic setting, we just focus on the maximal fidelity $f_j$ and its successful probability $p_j$ instead of the average fidelity $F$. In this context, the purified state is obtained if the purification process is successful; otherwise, the output state is discarded, and the process is repeated. Such state purification is also referred to as probabilistic purification. Ref.~\cite{fiuravsek2004optimal} presented an optimal probabilistic purification protocol and derived the corresponding optimal fidelity, which can be represented as the eigenvalue of a specific matrix based on Choi operators~\cite{choi1975completely,jamiolkowski1972linear}. 

However, solving the optimal fidelity~\cite{fiuravsek2004optimal} becomes increasingly challenging as the dimension of the quantum system or the number of noisy states grows. Also, the optimal probability as well as the explicit relationship between fidelity and probability remain unclear, while little is known about the realization of such kind of method using quantum circuits. A recent work by Childs et al.~\cite{childs2023streaming} further proposed a streaming purification protocol based on the swap test~\cite{buhrman2001quantum}. It treated the swap test procedure as a gadget and employed it recursively to produce a purer qudit. While this approach offers a novel perspective on purification, it poses challenges in terms of quantum memory or circuit depth requirements. As the recursive depth increases, maintaining some purified states with longer coherence times becomes increasingly more difficult, limiting the scalability of the protocol.
Thus, a flexible framework to process probabilistic purification tasks is urgently needed. 
Such a framework should take as input the noise channel $\cN$, the number of input states affected by this channel $n$,
and the desired successful probability $p$, and outputs a probabilistic purification protocol that can achieve the maximal fidelity $F_{\cN}(n,p)$. 
This flexibility is particularly relevant in certain scenarios where the number of available noisy states may vary or where the purification process needs to adapt to dynamic changes in the quantum system. Additionally, taking into account the trade-offs between fidelity and probability, we may sometimes need to study the specific protocols that can achieve maximal fidelity while ensuring a certain level of success probability. 

In this paper, we address the limitations of existing purification protocols by leveraging semidefinite programs (SDP), block encoding techniques, and parameterized quantum circuits (PQC). We first formalize a SDP framework to investigate the probabilistic purification protocols that achieve the highest fidelity with a certain successful probability and trade-offs between them. 
In particular, based on the SDP framework, we re-derive the pioneering purification protocol proposed in work~\cite{barenco1997stabilization}, which projects the $n$ copies of depolarizing noise state into its symmetric subspace, and present the analytical forms of fidelity and probability. 
Remarkably, we prove its optimality when $n=2$ and arbitrary dimension $d$, and numerically demonstrate its optimality when $n\leq8$ and $d\leq10$. In order to demonstrate the potential and flexibility of our framework, we investigate it in different regimes with practical noise models.
We further give a feasible approach to implement the optimal protocol via efficient quantum circuits. 
In particular, we present explicit quantum circuits for $3$ or $4$ noisy input states, respectively, which significantly reduce the number of ancilla qubits compared with previous works~\cite{barenco1997stabilization,laborde2024quantum,yang2024quantum}. The strategy we derive these two circuits can be used for circuits with more inputs, i.e., $n>4$. We further develop an algorithm to estimate the sample complexity of the optimal protocol, determining the required number of unknown noisy state copies to achieve a specific fidelity. 
Based on efficient circuits for purification, we propose a recursive purification protocol that generalizes the swap-test-based purification method. By allowing an arbitrary number of inputs $n\geq 2$, the protocol achieves the desired fidelity with a reduced recursive depth.
 Overall, our results advance the understanding of probabilistic purification and provide practical tools for its implementation in quantum systems.

\textbf{Structure of the paper.} 
% In Section~\ref{sec:preliminaries}, we present the essential preliminaries and notations used throughout this work. 
In Section~\ref{sec:opt_protocol}, we investigate the probabilistic purification protocols in detail, including the SDP framework for general cases, optimal protocols for depolarizing noises, and flexibility for other noise models. We further study practical and efficient circuits for the optimal protocol in Section~\ref{sec:implementation}. Finally, we propose algorithms and recursive protocols in terms of practically and effectiveness in Section~\ref{sec:experiments}, demonstrating advantage of our framework compared with existing methods.

\section{Optimal probabilistic purification protocol}\label{sec:opt_protocol}
In this section, we investigate probabilistic purification protocols and the optimality in different noise models. Before delving into the main subject, we introduce some notations. We consider a finite-dimensional Hilbert space $\cH_A$ and denote $A$ as the quantum system. Denote the dimension of $\cH_A$  as $d$.  Let $\{\ket{j} \}_{j=0,\cdots,d-1}$ be a standard computational basis. Denote $\cL(\cH_A)$ as the set of linear operators that map from $\cH_A$ to itself. A linear operator in $\cL(\cH_A)$ is called a density operator if it is positive semidefinite with trace one, and we denote $\cD(\cH_A)$ as the set of all density operators on $\cH_A$. A quantum channel $\cN_{A\to A}$ is a linear map from $\cL(\cH_A)$ to $\cL(\cH_{A})$ that is completely positive and trace-preserving (CPTP). Its associated Choi-Jamiołkowski operator is expressed as $J^{\cN}_{AA}:= \sum_{i, j=0}^{d-1}\ketbra{i}{j} \ox \cN_{A \to A}(\ketbra{i}{j})$. We denote $\cS_n$ as the symmetric group of degree $n$ and $\mathbf{P}_n(c)$ as the permutation operator where $c\in\cS_n$. $\Pi_n:=\frac{1}{n!}\sum_{c\in\cS_n}\mathbf{P}_n(c)$ denotes the projector on the symmetric subspace of $\cH_{A}^{\otimes n}$. 
% For example, if $n=2$ and $d=2$, we have the symmetric group $\cS_2:=\{(1),(12)\}$, $\mathbf{P}_2((1))$ and $\mathbf{P}_2((12))$ are identity operator and  the swap operator, respectively. 

Using multiple copies of the unknown noisy state, our objective is to recreate a qudit state with high fidelity to the ideal state for a certain probability. Specifically, given input noisy state $\rho^{\otimes n}\in\cD(\cH_A^{\otimes n})$, we focus on the completely positive trace non-increasing (CPTN) maps $\cE_{A^n\to A}$ that can generate an unnormalized state $\sigma^\prime\in\cL(\cH_A)$ such that the purified state $\sigma:=\sigma^\prime/\tr[\sigma^\prime]\in\cD(\cH_A )$ closely approximates an ideal pure state, with a certain probability. Notice that, in this work, the whole performance of purification protocols will be investigated, including the fidelity and probability over all noisy states affected by the same noise channel. 
% \honghao{We only consider the probability and fidelity right?}
The specific framework of probabilistic purification protocols is illustrated in Fig.~\ref{fig:framework}.

\begin{figure}[t]
    \centering
    \includegraphics[width=0.7\linewidth]{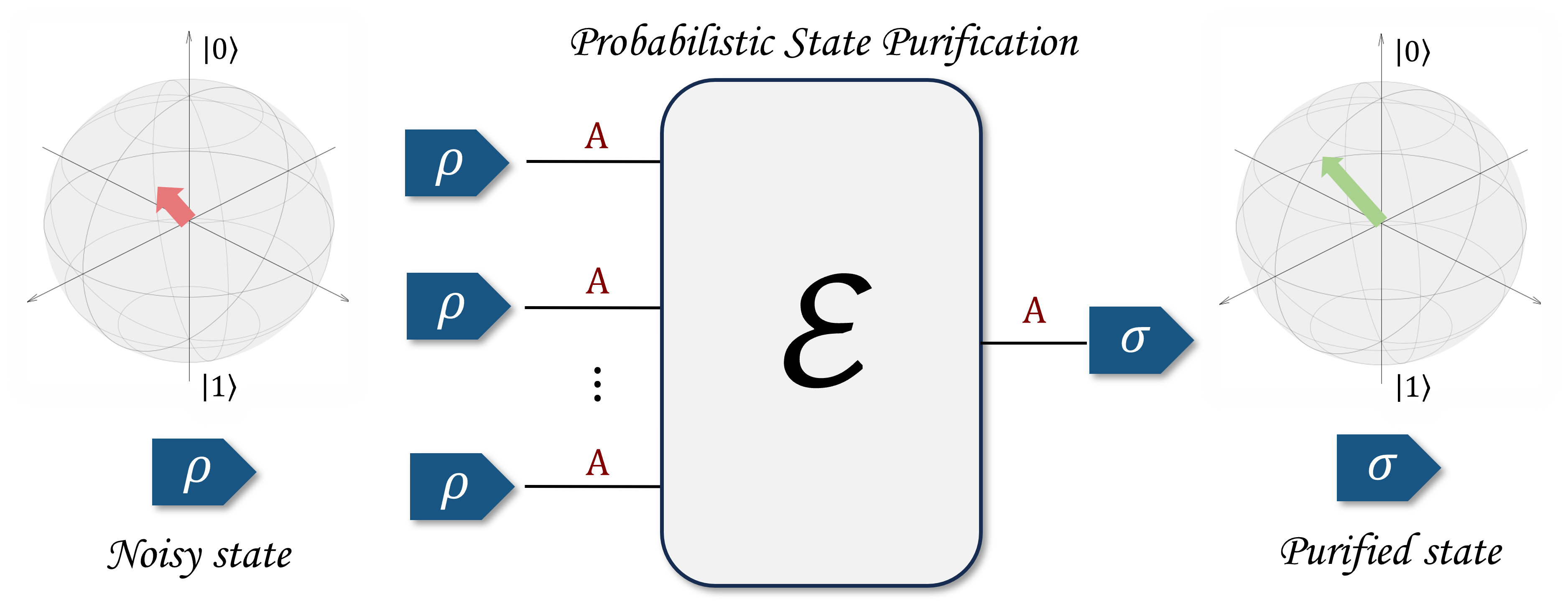}
         \caption{Framework of the probabilistic purification protocol. this protocol is designed to generate a purer state $\sigma$ by  $n$ copies of the noisy state $\rho$ with a certain probability $\tr[\sigma^\prime]$, where $\sigma^\prime:=\cE(\rho^{\ox n})$ and $\sigma:=\sigma^\prime/\tr[\sigma^\prime]$.}
    \label{fig:framework}
\end{figure}

Notice that probabilistic purification protocols are characterized by two key measures: the fidelity between the output state of purification protocols and the ideal pure state averaged over all state in Hilbert space, and the associated successful probability.  With this in mind, attention is directed towards those protocols that can achieve the highest fidelity for a given probability, as well as those that can attain the highest probability for a given fidelity. To formalize this, we introduce the following definitions:

% \begin{shaded}
\begin{definition}[Maximal universal fidelity for given probability]\label{def:optimal_fidelity}Let $\cN$ be the noise channel applied to system $A$ and $n$ be the number of copies of the noisy state. For any given probability $p\in(0,1)$, the maximal universal fidelity over all possible protocols is defined as
\begin{equation}
    \begin{aligned}
        F_{\cN}(n,p):=\max\left\{\frac{\int \tr[\sigma_{\psi}\psi]\,d\psi}{\int \tr[\sigma_{\psi}]\,d\psi}\Big\arrowvert\int\tr[\sigma_{\psi}]\,d\psi=p,\,\cE_{A^n\to  A}\in\rm{CPTN}\right\},
    \end{aligned}
\end{equation}
where $\sigma_{\psi}:=\cE_{A^n\to A}(\cN(\psi)^{\otimes n})$, and the integral is taken with respect to the Haar measure over all unit vectors in the Hilbert space $\cH_A$.
\end{definition}
% \end{shaded}
Similarly, the optimal purification probability is given by the following definition: 
% \begin{shaded}
\begin{definition}[Maximal universal probability for given fidelity]\label{def:optimal_probability}Let $\cN$ be the noise channel applied to system $A$ and $n$ be the number of copies of the noisy state. For any given fidelity $f\in(0,1)$, the maximal universal probability over all possible protocols is defined as
\begin{equation}
    \begin{aligned}
        P_{\cN}(n,f):=\max\left\{\int\tr[\sigma_{\psi}]\,d\psi\Big\arrowvert \frac{\int \tr[\sigma_{\psi}\psi]\,d\psi}{\int \tr[\sigma_{\psi}]\,d\psi}=f,\,\cE_{A^n\to A}\in\rm{CPTN}\right\},
    \end{aligned}
\end{equation}
where $\sigma_{\psi}:=\cE_{A^n\to A}(\cN(\psi)^{\otimes n})$, and the integral is taken with respect to the Haar measure over all unit vectors in the Hilbert space $\cH_A$.
\end{definition}
% \end{shaded}

A CPTN map is identified as the optimal probabilistic purification protocol for a given purification probability $p$ if it attains the maximum fidelity $F_{\cN}(n,p)$. Similarly, it is labeled as the optimal probabilistic purification protocol for a given purification fidelity $f$ if a CPTN map achieves the highest probability $P_{\cN}(n,f)$. These two measures provide a framework for a detailed exploration of probabilistic purification protocols and trade-offs between purification fidelity and probability.

%%%%%%%%%%%%%%%%%%%%%%%%%%%%%%%%%%%%%%%%%%%%%%%%%%%%%%%%%%%%%%%%%%%%%%%%%%%%
\subsection{Semidefinite programs}
We investigate probabilistic purification protocols using the semidefinite programs technology, which can be efficiently solved by the interior point method~\cite{boyd2004convex} with a runtime polynomial in system dimension $d$. Specifically, two SDPs are formalized to calculate the maximal fidelity and the probability, respectively. We demonstrate the former through the following result:
% \begin{tcolorbox}
\begin{proposition}\label{prop:primal_dual_sdp_fid}
Let $\cN$ be the noise channel applied to system $A$ and $n$ be the number of copies of an unknown noisy state. For a given purification probability $p\in(0,1)$, the maximal fidelity is given by the following SDP:
\begin{equation}\label{sdp:primal_dual_fid}
\begin{small}
\begin{aligned}
&\underline{\textbf{Primal Program}}\\
    F_{\cN}(n,p)=\max\;&\tr[J^{\cE}_{{A^nA}}Q^{T_{A^n}}_n]/p\\
    {\rm s.t.}\;& \tr[J^{\cE}_{{A^nA}}R^{T_{A^n}}_n]= p,\\
    &J^{\cE}_{A^nA}\geq 0,\,\tr_{A}[J_{A^n A}^{\cE}] \leq I_{A^n}
\end{aligned}
\end{small}
\quad
\begin{small}
\begin{aligned}
&\underline{\textbf{Dual Program}}\\
    \min\;& -x-\tr[Y_{A^n}]/p\\
    {\rm s.t.}\;& Q^{T_{A^n}}_n+R^{T_{A^n}}_n x+Y_{A^n}\otimes I_{A}\leq 0\\
    &Y_{A^n}\leq 0,
\end{aligned}
\end{small}
\end{equation}
where $J^{\cE}_{A^n A}$ denotes the Choi operator of CPTN map $\cE_{A^n\to A}$, $T_{A^n}$ denotes partial transpose on system $A^n$, $Q_n:=\cN^{\otimes n}\otimes id(\Pi_{n+1}/D(n+1,d))$, $R_n:=\cN^{\otimes n}(\Pi_n/D(n,d))\otimes I_{A}$, $D(n,d):=\tbinom{n+d-1}{n}$, and $\Pi_n$ is the projector on the symmetric subspace of $\cH^{\otimes n}_A$.
\end{proposition}
% \end{tcolorbox}
\begin{proof}
We demonstrate the objective function and the corresponding constraints of the primal SDP, respectively. By definition, the objective function of the primal SDP is as follows:
\begin{equation}
    \begin{aligned}
        \int \tr\left[\sigma_{\psi}\psi\right]d\psi&=\int \tr\left[\cE_{A^n\to A}(\cN(\psi)^{\otimes n})\psi\right]d\psi\\
        &=\int \tr\left[(J^\cE_{A^n  A})^{T_{A^n}}(\cN(\psi)^{\otimes n}\otimes \psi)\right] \,d\psi\\
        &=\tr\left[(J^{\cE}_{{A^n A}})^{T_{A^n}}\int\cN(\psi)^{\otimes n}\otimes\psi\,d\psi\right].
    \end{aligned}
\end{equation}
Using Schur lemma~\cite{harrow2013church,khatri2020principles}, that is the property, $\Pi_n=\tbinom{n+d-1}{n}\int \psi^{\otimes n}\,d\psi$, one can find that $\int\cN(\psi)^{\otimes n}\otimes\psi\,d\psi=\cN^{\otimes n}\otimes id(\Pi_{n+1}/D(n+1,d))$. Similarly, $\int\cN(\psi)^{\otimes n}\otimes I_A\,d\psi=\cN^{\otimes n}(\Pi_n/D(n,d))\otimes I_{A}$. Furthermore, the constraints derive from the fact that $\cE_{A^n\to A}$ is a CPTN map. Therefore, we complete this proof.
\end{proof}

The derivation about dual SDP is detailed in Appendix~\ref{appendix:dual_sdp}. The primal and dual SDPs proposed in Proposition~\ref{prop:primal_dual_sdp_fid} allow us to characterize the optimal probabilistic purification protocol for a given purification probability. Similarly, we present an SDP to calculate the maximal probability for a fixed purification fidelity.
\begin{equation}\label{sdp:primal_dual_prob}
\begin{small}
\begin{aligned}
&\underline{\textbf{Primal Program}}\\
    P_{\cN}(n,f)=\max\;&\tr[J^{\cE}_{{A^n A}}R^{T_{A^n}}_n]\\
    {\rm s.t.}\;& \tr[J^{\cE}_{{A^n A}}Q^{T_{A^n}}_n]=\tr[J^{\cE}_{{A^n A}}R^{T_{A^n}}_n]f,\\
    &J^{\cE}_{A^n A}\geq 0,\,\tr_{ A}[J_{A^n A}^{\cE}] \leq I_{A^n}
\end{aligned}
\end{small}
\quad
\begin{small}
\begin{aligned}
&\underline{\textbf{Dual Program}}\\
    \min\;& -\tr[Y_{A^n}]\\
    {\rm s.t.}\;&(1-xf)R^{T_{A^n}}_n+Q^{T_{A^n}}_n x+
     Y_{A^n}\otimes I_{A} \leq 0\\
    &Y_{A^n}\leq 0,
\end{aligned}
\end{small}
\end{equation}
A derivation can be found in Appendix~\ref{appendix:dual_sdp}. There is a similar SDP proposed in~\cite{fiuravsek2004optimal}, which can obtain the maximal probability for a given maximal fidelity. However,  we have extended this result in this work. The SDP in Eq.~\eqref{sdp:primal_dual_prob} allows us to calculate the maximal probability for any given purification fidelity, which means that one can observe the probability trend as purification fidelity changes. In summary, the SDPs in Eq.~\eqref{sdp:primal_dual_fid} and Eq.~\eqref{sdp:primal_dual_prob} facilitate our exploration of the trade-off between the fidelity and probability. 

It is worth emphasizing that the versatility of our SDPs enables them to operate effectively across a wide range of noise channels. Moreover, their performance becomes even more remarkable when focusing on a particular noise channel. In the subsequent section, we delve into the speciﬁcs.

%%%%%%%%%%%%%%%%%%%%%%%%%%%%%%%%%%%%%%%%%%%%%%%%%%%%%%%%%%%%%%%%%%%%%%%%%%%%

\subsection{Depolarizing noise}
We focus on the depolarizing noise, which can describe average noise for large-scale circuits involving many qudits and gates~\cite{urbanek2021mitigating}. It is crucial for quantum computation and quantum communication to overcome this noise. Utilizing the SDPs in Eq.~\eqref{sdp:primal_dual_fid} and Eq.~\eqref{sdp:primal_dual_prob}, the explicit purification protocol, the fidelity and the corresponding probability are explored. 

For the $n=2$ case, which means two copies of a $d$-dimensional noisy state as the input states for purification protocols, we will demonstrate that there exists a purification protocol that can obtain the maximal fidelity $f_2$ with probability $p_2$. The probability $p_2$ is maximal over all protocols that can achieve the fidelity $f_2$. Furthermore, we generalize the purification protocol to the $n>2$ case and demonstrate its optimality by numerical experiments. The following conclusion characterizes the purification protocol for the $n=2$ case.

% \begin{tcolorbox}
\begin{theorem}[Golden point of probabilistic purification]\label{thm:opt_fidelity_probability}
Let $\delta$ be the noise parameter of the depolarizing channel $\cN_\delta$, and $\Lambda:=\operatorname{Diag}(\lambda_0\,\cdots,\lambda_{d-1})$ be the eigenvalue matrix of a $d$-dimensional noisy state, where  $\lambda_0:=1-\frac{d-1}{d}\delta$, $\lambda_j:=\frac{\delta}{d}$, for $j=1,\cdots,d-1$. Then, we have
\begin{equation}
    \begin{aligned}
        F_{Depo}(2,p_2)=f_2,\quad P_{Depo}(2,f_2)=p_2,
    \end{aligned}
\end{equation}
where $p_2 := \frac{1}{2}(1+\tr[\Lambda^2])$, $f_2:=\frac{1}{2p_2}(\lambda_0+\lambda_0^2)$.
\end{theorem}
% \end{tcolorbox}
\begin{proof}
Firstly, we will show $f_2 \leq F_{Depo}(2,p_2)$ using the primal SDP in Eq.~\eqref{sdp:primal_dual_fid}. Specifically, it is straightforward to see that the following Choi operator is a feasible solution:
\begin{equation}
    \begin{aligned}
        J^{\cE}_{A^2A}:=\tr_{A_{O}}[J^{\cM_2}_{A_{I}^2A_{O}^2}],
    \end{aligned}
\end{equation}
where $J^{\cM_2}_{A_{I}^2A_{O}^2}$ denotes the Choi operator of CPTN map $\cM_2(\cdot)$, $\cM_2(\cdot):=\Pi_2(\cdot)\Pi^\dagger_2$ and $\Pi_2$ is the projector on the symmetric subspace of $\cH^{\otimes 2}_A$. $A_{I}$ and $A_{O}$ are labeled for input  and output system $A$ to distinguish them. Then, one can obtain the objective value $f_2$ using this solution, as follows:
\begin{equation}
    \begin{aligned}
        \tr\left[J^{\cE}_{A^2 A}Q_2^{T_{12}}\right]&=\tr\left[(J^{\cE}_{A^2 A})^{T_{12}}\cN^{\otimes 2}_\delta\otimes id(\Pi_{3}/D(3,d))\right]\\
        &=\int \tr\left[(J^{\cE}_{A^2A})^{T_{12}}\cN_\delta(\psi)^{\otimes 2}\otimes \psi\right] \,d\psi\\
        &=\int \tr\left[\psi\tr_{2}\left[\Pi_2\cN_\delta(\psi)^{\otimes 2}\Pi_2^\dagger\right]\right]\,d\psi\\
        &=\bra{0}\tr_{2}\left[\Pi_2\Lambda^{\otimes 2}\Pi_2^\dagger\right]\ket{0},
    \end{aligned}
\end{equation}
where $T_{12}$ denotes the partial transpose on system $A^2$, and the last equality follows from the fact that $\bra{0}\tr_{2}[\Pi_2\Lambda^{\otimes 2}\Pi_2^\dagger]\ket{0}=\bra{\psi}\tr_{2}[\Pi_2\cN_\delta(\psi)^{\otimes 2}\Pi_2^\dagger]\ket{\psi}$ for any pure state $\psi$. Furthermore, we have $\tr[J^{\cE}_{A^2 A}R_2^{T_{12}}]= \tr[\Pi_2\Lambda^{\otimes 2}]$ in the same way. According to Lemma~\ref{lem:recursion_expression}, we further obtain the recursive forms of probability and fidelity, which means $F_{Depo}(2,p_2)\geq f_2$.

Second, we demonstrate that $\{-f_2,\mathbf{0}\}$ is a feasible solution of the dual SDP in Eq.~\eqref{sdp:primal_dual_fid}. To be specific, we observe the constraint as follows:
\begin{equation}
    \begin{aligned}
    f_2R^{T_{12}}_2 - Q_2^{T_{12}}\geq 0.
    \end{aligned}
\end{equation}
Notice that it is equivalent to examining the positivity of the operator $f_2R_2-Q_2^{T_3}$. According to Lemma~\ref{lem:case_of_two}, we have $f_2R_2-Q_2^{T_3}\geq 0$, which implies $F_{Depo}(2,p_2)\leq f_2$. Combining the primal part and the dual part, we conclusively establish that $F_{Depo}(2,p_2)=f_2$.

Similarly, one can find that the $J^{\cE}_{A^2 A}$ and $\{+\infty,-\Pi_2p_2/D(2,d)\}$ are the feasible solutions of the primal and dual SDP in Eq.~\eqref{sdp:primal_dual_prob} for calculating $P_{Depo}(2,f_2)$, respectively, where $D(2,d):=\frac{(d+1)d}{2}$. Notably, for a sufficiently large $x$, the following inequality is satisfied: 
\begin{equation}
    \begin{aligned}
        (1-xf_2)R^{T_{12}}_2+Q_2^{T_{12}}x+ Y_{A^2}\otimes I_{A} \leq 0,
    \end{aligned}
\end{equation}
which is equivalent to $f_2R^{T_{12}}_2-Q_2^{T_{12}}\geq 0$. Similarly, combining the primal and dual parts, we have $P_{Depo}(2,f_2)=p_2$. Thus, we complete the proof.
\end{proof}

Theorem \ref{thm:opt_fidelity_probability} elucidates the tight connection between $f_2$ and $p_2$. Specifically, when the success probability $p_2$ is given, the solution to the SDP problem in Eq.~\eqref{sdp:primal_dual_fid} is uniquely determined as $f_2$. Similarly, given the value of $f_2$, the corresponding probability $p_2$  can be obtained from the solution of the SDP problem in Eq.~\eqref{sdp:primal_dual_prob}. Notice that the latter clarifies that the probability $p_2$ is maximal over all protocols that can achieve the fidelity $f_2$. However, a natural issue is whether there exists a purification protocol that can obtain higher fidelity with a certain success probability $p$, i.e., $F_{Depo}(2,p)>f_2$. We answer this question using the following conclusion:
% \begin{tcolorbox}
\begin{proposition}[Maximal universal fidelity]\label{prop:opt_fidelity}
Let $\delta$ be the noise parameter of the depolarizing channel, and $\Lambda:=\operatorname{Diag}(\lambda_0\,\cdots,\lambda_{d-1})$ be the eigenvalue matrix of a $d$-dimensional noisy state, where  $\lambda_0:=1-\frac{d-1}{d}\delta$, $\lambda_j:=\frac{\delta}{d}$, for $j=1,\cdots,d-1$. Then, the maximal fidelity across all possible probabilities is given by
\begin{equation}
    \begin{aligned}
        f_2=\max_{\,p\in(0,1)}F_{Depo}(2,p),
    \end{aligned}
\end{equation}
where $f_2:=\frac{1}{2p_2}(\lambda_0+\lambda_0^2)$ and $p_2 := \frac{1}{2}(1+\tr[\Lambda^2])$.
\end{proposition}
% \end{tcolorbox}
\begin{proof}
    Firstly, we will demonstrate the purification fidelity $F_{Depo}(2,p)$ is a decreasing function with respect to the purification probability $p\in(0,1)$, i.e., $F_{Depo}(2,p)\geq F_{Depo}(2,q)$ for $p\leq q \in(0,1)$. Specifically, let $J_q$ be the optimal solution of the primal SDP in Eq.~\eqref{sdp:primal_dual_fid} for fixing purification probability $q$, that is, $F_{Depo}(2,q)=\tr[J_qQ_2^{T_{12}}]/q$. Then, the Choi operator $J_p:=\frac{p}{q}J_q$ is a feasible solution of the primal SDP in Eq.~\eqref{sdp:primal_dual_fid} for fixing purification probability $p$, which means 
    \begin{equation}
            \begin{aligned}
        F_{Depo}(2,p)\geq\tr[J_pQ_2^{T_{12}}]/p=\tr[J_qQ_2^{T_{12}}]/q=F_{Depo}(2,q).
    \end{aligned}
    \end{equation}
    
    Secondly, we demonstrate that even if the purification probability is reduced, it will not increase the purification fidelity, that is, $F_{Depo}(2,p)=f_2$ for all $p\leq p_2$. In specific, one can check that the Choi operator $J_p:=\frac{p}{p_2}J^{\cE}_{A^2A}$ is the optimal solution of the primal SDP in Eq.~\eqref{sdp:primal_dual_fid} with respects to probability $p$, where the definition of $J^{\cE}_{A^2 A}$ is proposed in the proof of Theorem~\ref{thm:opt_fidelity_probability}. Therefore, we obtain that $f_2$ is the maximal purification fidelity for all given probabilities.
\end{proof}

Theorem~\ref{thm:opt_fidelity_probability} and Proposition~\ref{prop:opt_fidelity} imply that $\cE_{A^2\to A}^\ast:=\tr_2\circ\cM_2$ is the optimal protocol that can achieve the \textit{golden point}, i.e., the fidelity $f_2$ is optimal and the corresponding probability $p_2$ is also maximal for all probabilistic protocols that can obtain the fidelity $f_2$.
\begin{figure}[t]
    \centering
    \subfigure[]{
    \includegraphics[width=0.45\textwidth]{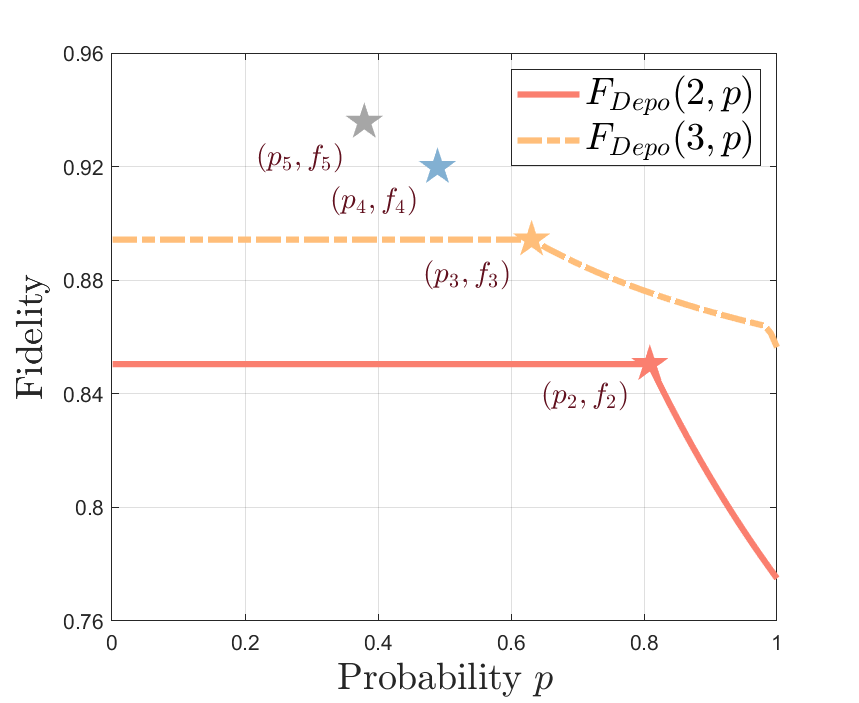}
    \label{fig:f_P_relationship}}
    \subfigure[]{
    \includegraphics[width=0.48\textwidth]{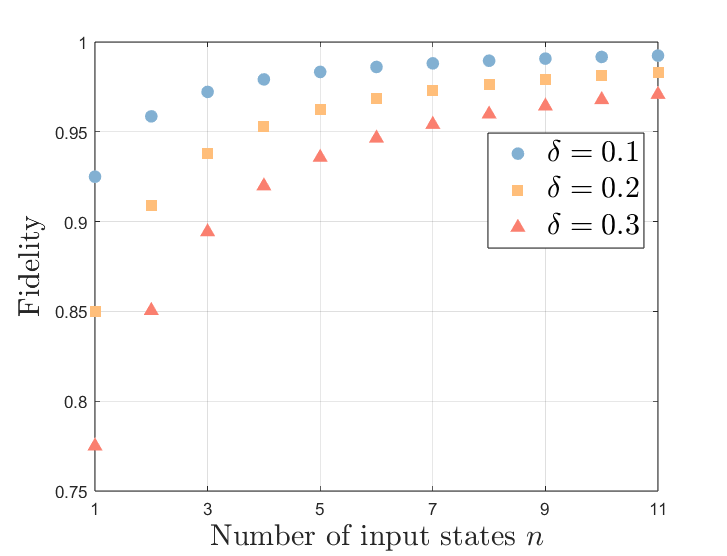}
    \label{fig:fn_n_delta}}
    \caption{Universal purification fidelity $F_{\text{Depo}}(n,p)$ with different number of input depolarized states $n$ and success probability $p$. (a) describes the relation between the general protocol $\cE^\ast_{A^n\to  A}$ in Remark~\ref{rem:optimal_protocol} and the protocols obtained from the SDP in Eq.\eqref{sdp:primal_dual_fid} when the error parameter $\delta$ is $0.3$ and the dimension $d=4$. The star-markers denote the protocol $\cE^\ast_{A^n\to A}$ follow from the probability $p_n := \frac{1}{n}\sum_{j=1}^n p_{n-j} \tr\left[ \Lambda^j\right]$ and fidelity $f_n:=\frac{1}{n}\sum_{j = 1}^np_{n-j} \lambda_0^j/p_n$. (b) describes the trend of the fidelity $f_n$ with respect to the number of copies of a $4$-dimensional input state for different noise parameters $\delta$.}
    \label{fig:f_P_delta_relationship}
\end{figure}

% \begin{shaded}
\begin{remark}(\textbf{General purification protocol})\label{rem:optimal_protocol}~A protocol for the general $n$ is further proposed as follows:
\begin{equation} 
\begin{aligned}
    \cE^\ast_{A^n\to A}:=\tr_{2\cdots n}\circ\cM_n,
\end{aligned}
\end{equation}
where $\tr_{2\cdots n}$ denotes the partial trace on the subsystems $A_2\cdots A_n$, $\cM_n(\cdot):=\Pi_n(\cdot)\Pi^\dagger_n$ and $\Pi_n$ is the projector on the symmetric subspace of $\cH^{\otimes n}_A$. Furthermore, for the depolarizing noise, the explicit forms of the maximal fidelity $f_n$ and successful probability $p_n$ are presented. See Lemma~\ref{lem:recursion_expression} for details. 
\end{remark}
% \end{shaded}

It should be noted that this general protocol involves an operation of projection into the symmetric subspace of $n$ copies of noisy state~\cite{barenco1997stabilization}, while a discussion regarding its optimality has not been provided. Based on the observation in the case of $n=2$, we conjecture that the general probabilistic purification protocol $\cE^\ast_{A^n\to A}$ is optimal. We propose a strategy to prove its optimality by using primal and dual SDP in Eq.\eqref{sdp:primal_dual_fid}. Specifically, it can be simplified to a more straightforward problem in which we only need to check the positivity of the operator $f_nR_n^{T_{A^n}}-Q_n^{T_{A^n}}$. Inspired by the case of $n=2$ proved in Theorem~\ref{thm:opt_fidelity_probability} and Proposition~\ref{prop:opt_fidelity}, one can represent this operator as a linear combination of operators $\mathbf{P}_n(c)^{T_{n}}$, $c\in\cS_n$, and decompose it into a specific basis based on the irreducible representations of the symmetric group $\cS_n$. Instead of furnishing proof for general cases, the optimality of the protocol is established for the case of $n\leq 8$ and $d\leq 10$ through numerical experiments shown in Table~\ref{tab:comparison_n_d}.

\begin{table}[H]
\centering
\setlength{\tabcolsep}{0.8em}
\resizebox{0.9\textwidth}{!}{
\begin{threeparttable}
\begin{tabular}{l|cccccccccccc}
\toprule
\midrule
\addlinespace
& $ d = 2{~\cite{cirac1999optimal}}$ 
& $ d = 3{~\cite{fu2016quantum}}$
& $ d = 4$ 
& $ d = 5$
& $ d = 6$
& $ d = 7$
& $ d = 8$
& $ d = 9$
& $ d = 10$
& $\cdots$
& $d = \infty$
\\
\midrule
\addlinespace
$ n = 2$
& $\usym{2713}$ $\checked$
& $\usym{2713}$ $\checked$
& $\usym{2713}$ $\checked$
& $\usym{2713}$ $\checked$
& $\usym{2713}$ $\checked$
& $\usym{2713}$ $\checked$
& $\usym{2713}$ $\checked$
& $\usym{2713}$ $\checked$
& $\usym{2713}$ $\checked$
& $\usym{2713}$ 
& $\usym{2713}$ 
\\
\midrule
\addlinespace
$ n = 3$
& $\usym{2713}$ $\checked$
& $\usym{2713}$ $\checked$
& $\checked$
& $\checked$
& $\checked$
& $\checked$
& $\checked$
& $\checked$
& $\checked$
& 
& 
\\
\midrule
\addlinespace
$ n = 4 $
& $\usym{2713}$ $\checked$
& $\usym{2713}$ $\checked$
& $\checked$
& $\checked$
& $\checked$
& $\checked$
& $\checked$
& $\checked$
& 
& 
& 
\\
\midrule
\addlinespace
$ n = 5 $
& $\usym{2713}$ $\checked$
& $\usym{2713}$ $\checked$
& $\checked$
& $\checked$
& $\checked$
& 
& 
& 
& 
& 
& 
\\
\midrule
\addlinespace
$ n =6 $
& $\usym{2713}$ $\checked$
& $\usym{2713}$ $\checked$
& $\checked$
& $\checked$
& 
& 
& 
&  
& 
& 
& 
\\
\midrule
\addlinespace
$ n =7 $
& $\usym{2713}$ $\checked$
& $\usym{2713}$ $\checked$
& $\checked$
&  
& 
& 
& 
&  
& 
& 
& 
\\
\midrule
\addlinespace
$ n = 8$
& $\usym{2713}$ $\checked$
& $\usym{2713}$ $\checked$
& 
&  
& 
& 
& 
&  
& 
& 
& 
\\
\midrule
\addlinespace
$\vdots$
& $\usym{2713}$
& $\usym{2713}$
& 
&  
& 
& 
& 
&  
& 
& 
& 
\\
\midrule
\addlinespace
$\infty$
& $\usym{2713}$
& $\usym{2713}$

\\
\midrule
\bottomrule
\end{tabular}

\begin{tablenotes} 
\item  $\quad$ $\usym{2713}$: theory-proved
cases $\quad$ $\checked$: experiment-verified cases
\end{tablenotes} 
\end{threeparttable}}
\caption{Optimal state purification protocols proven by theory or verified through experiments for $n$ copies of an unknown $d$-dimensional state. 
The thick check marks and the thin check marks represent theory-proved cases and experiment-verified cases, respectively.
In $d = 2$ and $d = 3$ cases, our general protocol's 
results match the theoretical value proved by~\cite{cirac1999optimal} and~\cite{fu2016quantum}, respectively. Furthermore, we prove the optimality of our protocol for the case of $n=2$ and arbitrary dimension $d$ by using Theorem~\ref{thm:opt_fidelity_probability} and Proposition~\ref{prop:opt_fidelity}. In addition, due to the limitation of computer hardware, we demonstrate the optimality for $n\leq 8$ and $d\leq 10$ through numerical experiments.}
\label{tab:comparison_n_d}
\end{table}

It is worth noting that the authors made a similar observation in ~\cite{fiuravsek2004optimal}. However, we demonstrate it differently and propose more explicit forms encompassing the purification protocol, fidelity, and probability. In specific, from Fig.~\ref{fig:f_P_relationship}, one can observe that the probabilistic purification protocols involve a trade-off between fidelity and probability, and there is an inflection point where the universal fidelity is highest and the fidelity will decrease after this point, which we call the \textit{golden point} for probabilistic purification. Remarkably, for depolarizing noise, there exists an explicit protocol that achieves this golden point. Furthermore, one can calculate the maximal fidelity $f_n$ and the corresponding probability $p_n$ directly depending on their recursive forms instead of solving SDPs, which means one can analyze the trade-off between fidelity and probability for large values of $n$. In addition, Fig.~\ref{fig:fn_n_delta} depicts the trend that the maximal fidelity $f_n$ will converge to one as the number of input states $n$ increases for various error parameters. 

%%%%%%%%%%%%%%%%%%%%%%%%%%%%%%%%%%%%%%%%%%%%%%%%%%%%%%%%%%%%%%%%%%%%%%%%%%
%%%%%%%%%%%%%%%%%%%%%%%%%%%%%%%%%%%%%%%%%%%%%%%%%%%%%%%%%%%%%%%%%%%%%%%%%%
\subsection{Pauli noise and amplitude damping noise}\label{sec:Comparison experiments} 
In this section, our objective is to elucidate the feasibility of our SDP framework across various noise conditions. This analysis enables a comparative understanding of the performance disparities between protocols predicated on the SDP framework and those established protocols, including the symmetric projection strategy $\cE^\ast_{A^n\to A}$.

We focus on the amplitude damping noise~\cite{nielsen2001quantum} and Pauli noise~\cite{erhard2019characterizing} with different settings. Specifically, we depict lines and stars with respect to different noise parameter $\delta$ in each subplot in Fig.~\ref{fig:sdp_comparison}, where lines denote the protocols based on the optimal solution of SDP in Eq.~\eqref{sdp:primal_dual_fid} and stars symbolize the protocol based on the swap test procedure (namely, the protocol $\cE^\ast_{A^2\to A}$). The comparison results shown in Fig.~\ref{Fig.SDP_pauli_3copies} and Fig.~\ref{Fig.SDP_pauli} reveal that the swap test procedure approximates optimality in purifying Pauli noise, especially for the case of two copies of noisy state as the input. By contrast, if $n=4$, there exist better purification protocols. Remarkably, for the amplitude damping noise, there exists a significant gap between the swap test protocol and the optimal one obtained from the solution of SDP, which means a more powerful purification protocol is needed to bridge this gap. 

In conclusion, by taking advantage of the SDP framework, one can optimize a purification protocol for given various input parameters, such as the noise types, the desired successful probability, and the amount of noisy states that are consumed per execution, which allows us to find and design a proper purification protocol in specific settings.

\begin{figure}[H]
\centering  
\subfigure[]{
\includegraphics[width=0.31\textwidth]{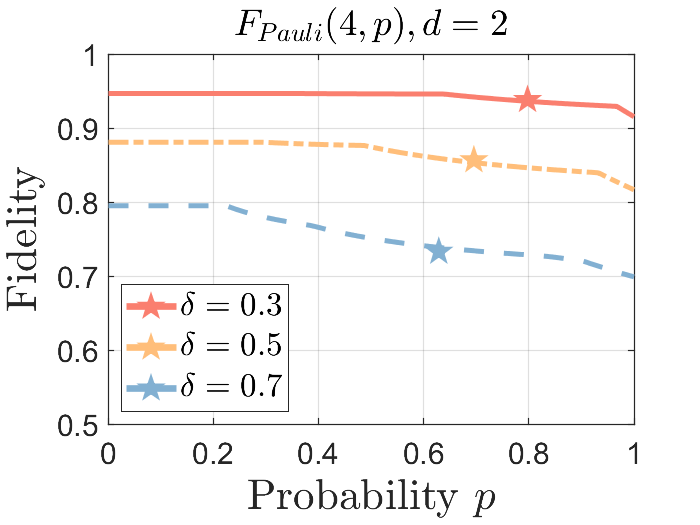}
\label{Fig.SDP_pauli_3copies}}
\subfigure[]{
\includegraphics[width=0.31\textwidth]{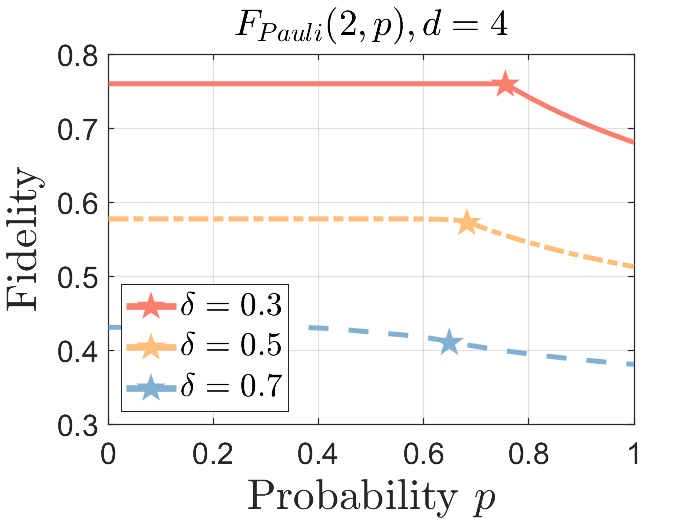}
\label{Fig.SDP_pauli}}
\subfigure[]{
\includegraphics[width=0.31\textwidth]{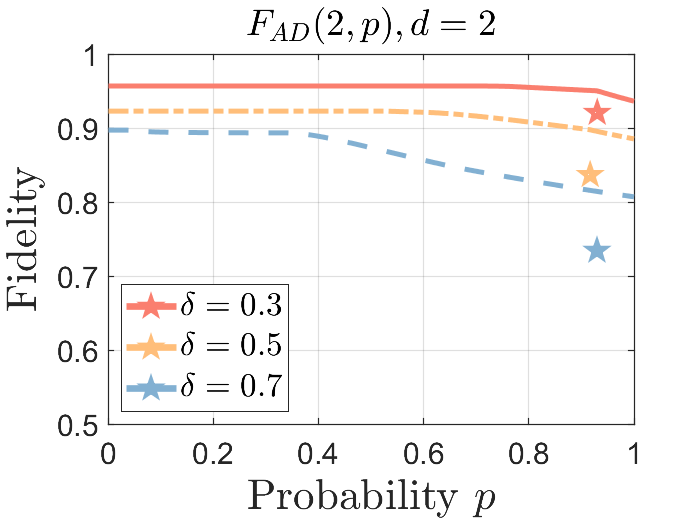}
\label{Fig.SDP_AD}}
\caption{Comparison of purification protocols obtained by SDP and the protocol based on the swap test for different noise channels. The lines describe the relationship between the success probability $p$ and the corresponding fidelity obtained by SDP in Eq.~\eqref{sdp:primal_dual_fid}. The star markers describe the results obtained by the protocol based on the swap test. (a) considers $4$-copy qubits affected by Pauli noise, with Kraus operators $K_0 := \sqrt{1-0.75\delta}I$, $K_1 := \sqrt{0.1\delta}X$, $K_2 := \sqrt{0.2\delta}Y$ and $K_3 := \sqrt{0.45\delta}Z$.
(b) addresses $2$-copies qudits ($d=4$) under the same Pauli noise, where the Kraus operators are $K_i := K_j \otimes K_k$ with $K_j, K_k \in \{\sqrt{1-0.75\delta}I, \sqrt{0.1\delta}X, \sqrt{0.2\delta}Y, \sqrt{0.45\delta}Z\}$. (c) describes $2$-copies qubits exposed to amplitude damping noise, with Kraus operators $K_0=\ketbra{0}{0}+\sqrt{1-\delta}\ketbra{1}{1}$ and $K_1=\sqrt{\delta}\ketbra{0}{1}$.
% $
% K_0 = \begin{pmatrix}
% 1 & 0 \\
% 0 & \sqrt{1-\delta}
% \end{pmatrix} $ and $
% K_1 = \begin{pmatrix}
% 0 & \sqrt{\delta} \\
% 0 & 0
% \end{pmatrix}
% $ .
}
\label{fig:sdp_comparison}
\end{figure}

%%%%%%%%%%%%%%%%%%%%%%%%%%%%%%%%%%%%%%%%%%%%%%%%%%%%%%%%%%%%%%%%%%%%%%%%%%%%
\section{Practical circuits based on block encoding and PQCs}\label{sec:implementation}
%%%%%%%%%%%%%%%%%%%%%%%%%%%%%%%%%%%%%%%%%%%%%%%%%%%%%%%%%%%%%%%%%%%%%%%%%%%%
% \honghao{Practicality and efficiency}
We further delve into the circuit realization of the probabilistic purification protocol $\cE^\ast_{A^n\to A}$ referenced in Remark~\ref{rem:optimal_protocol}, which focuses on projecting the noisy state $\rho^{\otimes n}$ into its symmetric subspace. Recognizing that each permutation element $c$ in the symmetric group $\mathcal{S}_n$ can be expressed as a product of two-cycle permutations, various studies~\cite{alicki1988symmetry,barenco1997stabilization} have introduced efficient quantum circuits for executing this projection operation. 
This is achieved by recursively managing the permutation operator $\mathbf{P}_n(c)$ and preparing an amplitude superposition of ancilla qubits. Nonetheless, the requirement for ancilla systems still scales polynomially with the number of copies of a noise state $n$. 

In practical terms, minimizing the ancilla usage becomes a critical consideration. To address this, we explore the integration of block encoding~\cite{low2019hamiltonian,gilyen2019quantum,camps2024explicit} and parameterized quantum circuits~\cite{benedetti2019parameterized} as innovative strategies to refine the understanding and implementation of the symmetric projection circuit. By leveraging these techniques, we aim to streamline the circuit design from a practical standpoint, potentially reducing the reliance on ancillary systems while maintaining operational efficiency and effectiveness.

Notice that the projector $\Pi_n$ is a linear combination of unitaries, i.e., $\Pi_n=\sum_{c\in\cS_n}\frac{1}{n!}\mathbf{P}_{n}(c)$, where $\mathbf{P}_{n}(c)$ is the permutation operator of the symmetric group $\cS_n$ of degree $n$. 
Thus, one can first block encode this projector $\Pi_n$ by the \textit{linear combinations of unitaries} (LCU) algorithm~\cite{childs2012hamiltonian,loaiza2023reducing}, i.e., constructing a quantum circuit $U$ for the state purification task, such that
\begin{equation}\label{eq:Unitary_bolck_encoding}
    \begin{aligned}
        U=\left[\begin{array}{cc}
            \Pi_n &  \cdot\\
            \cdot & \cdot
        \end{array}\right]
    \end{aligned}.
\end{equation} 
It consists of the preparation circuit and selection circuit. The former can produce a superposition state determined by coefficients of unitaries, while the latter function is to select the corresponding unitary based on ancilla systems. 
The following result demonstrates the implementation of the probabilistic purification protocol $\cE^\ast_{A^n\to A}$ within a quantum circuit.
% \begin{tcolorbox}
\begin{proposition}[Quantum circuit of purification protocol]
\label{prop:circuit of protocol} Let $\cE^\ast_{A^n\to  A}$ be the probabilistic purification protocol shown in Remark~\ref{rem:optimal_protocol}. Then, there exists a quantum circuit $U_{n\to1}:=(V^\dagger\otimes I_n) W (V\otimes I_n)$ such that
\begin{equation}\label{eq.cE}
    \begin{aligned}
        \cE^\ast_{A^n\to A}(\cdot)=\tr_{2\cdots n}\left[ \tr_{A^\prime}[(\ketbra{0}{0}_{A^\prime}\otimes I_n)U_{n\to1}(\ketbra{0}{0}_{A^\prime}\otimes \cdot)U^\dagger_{n\to1}(\ketbra{0}{0}_{A^\prime}\otimes I_n)]\right],
    \end{aligned}
\end{equation}
where $A^{\prime}$ is the the auxiliary
system, $I_n$ is the identity matrix of the subsystems $A_1\cdots A_n$, the preparation circuit $V$ satisfies $V\ket{0}_{A^\prime}=\sum_{k=0}^{n!-1}\frac{1}{\sqrt{n!}}\ket{k}$, and the selection circuit $W$ consists of $\lceil\log_d(n!-1)\rceil$ different controlled-$\mathbf{P}_n(c)$ gates, where $\mathbf{P}_n(c)$ is permutation operator and $\cS_n$ denotes the symmetric group of degree $n$.
\end{proposition}
% \end{tcolorbox}
% \honghao{we can argue it has efficient implementation poly(n, logd)-sized circuit}
Proposition~\ref{prop:circuit of protocol} implies that for a given $d$-dimensional depolarized state $\rho^{\otimes n}$, one can reconstruct a state $\sigma$ with high purity by employing a specific quantum circuit. For the sake of simplicity, we illustrate the circuit in this case when considering $n=3$ and $d=2$. The specific circuit is shown in Fig.~\ref{fig:protocol_3_1_v2}. In addition, one can find that the purification circuit can be reduced to the \textit{swap test gadget} if $n=2$.

\begin{figure}[H]
    \centering
    \includegraphics[width=0.72\textwidth]{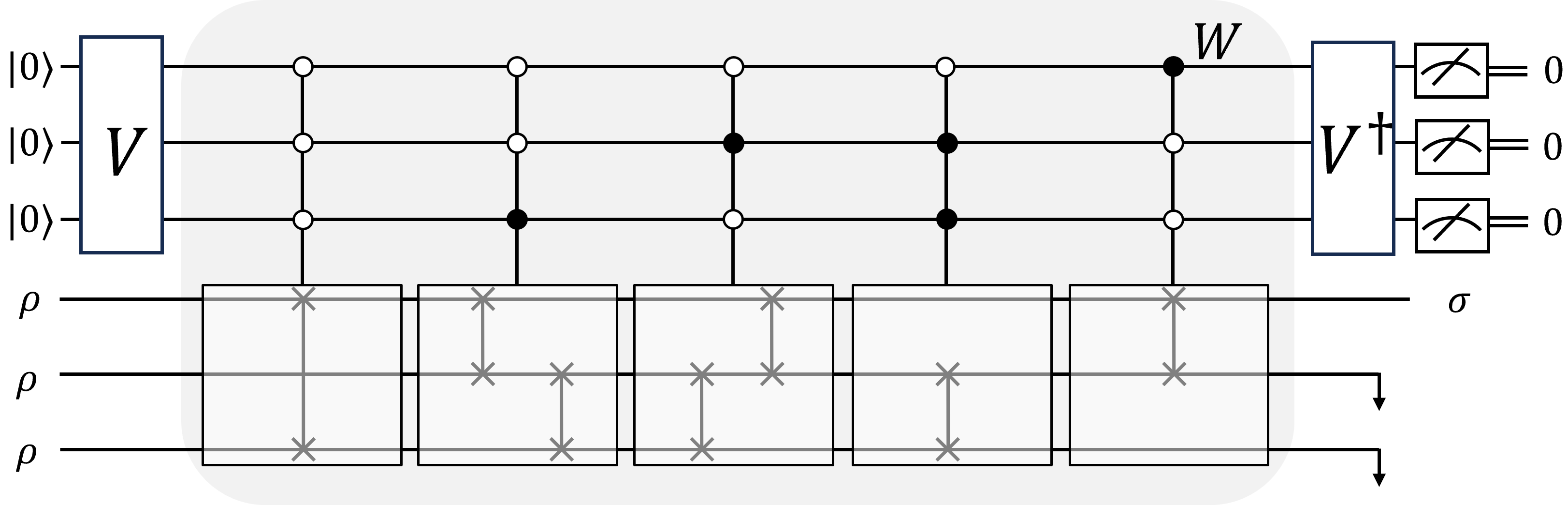}
    \caption{The purification circuit for three input noisy qubit states. The preparation circuit $V$ satisfies $V\ket{0}=\sum_{j=0}^{5}\frac{1}{\sqrt{6}}\ket{j}$, and the selection circuit $W$ consists of five different controlled-$\mathbf{P}_3(c)$ gates, where $c\in\{(12),(13),(23),(123),(132)\}$.}
    \label{fig:protocol_3_1_v2}
\end{figure}

% \honghao{more practical implementation}
The circuit outlined in Proposition~\ref{prop:circuit of protocol}, as well as the protocol proposed in prior works~\cite{barenco1997stabilization,laborde2024quantum}, indeed involve the use of linearly many ancilla systems, which could limit their practical applications due to resource constraints inherent in current quantum computing technologies. A promising direction to circumvent this limitation involves a strategic reduction in the reliance on ancillary systems. A novel approach in this context is the use of rotation gates, specifically $R_y$ gates, to replace the selection circuit applied to the ancilla system. 
When we use fewer ancilla qubits and $R_y$ gates on them to simulate the selection circuit, it will introduce products of permutation operators in the overall isometry. For example, see the proof of Proposition~\ref{prop:3qubit_c} below. 
Therefore, we need to carefully select the unitaries used in linear combinations, rather than employing all $n!$ permutation operators directly.
To be specific, we present explicit purification circuits for the cases of $n=3$ and $n=4$ in the following discussion.

% \begin{tcolorbox}
\begin{proposition}[Simplified purification circuit] 
\label{prop:3qubit_c}
Considering the purification protocol $\cE^\ast_{A^3\to  A}$ with three copies of noise qubit state $\rho\in\cD(\cH_{A})$, there exists a simplified quantum circuit $U_{3\to1}$ to realize it as follows:
    \begin{equation}
        \begin{aligned}
        U_{3\to1}:= \left(R_y(\gamma) \otimes I_3\right)\, C\mbox{-}\mathbf{P}_3((132))\, \left(R_y(\beta) \otimes I_3\right) \, C\mbox{-}\mathbf{P}_3((123))\,\left(R_y(\alpha) \otimes I_3\right)
        \end{aligned}
    \end{equation}
    where $\alpha=\gamma=-\arctan(\sqrt{2})$ and $\beta=\arccos(-1/3)$, $C\mbox{-}\mathbf{P}_3((132))$ and $C\mbox{-}\mathbf{P}_3((123))$ denote the controlled-$\mathbf{P}_3((132))$ and controlled-$\mathbf{P}_3((123))$, respectively.
\end{proposition}
% \end{tcolorbox}
\begin{proof}
Observing that for any quantum state $\rho\in\cD(\cH_{A})$, we have $\cE^\ast_{A^3\to  A}(\rho^{\otimes 3})=\tr_2[\Pi_3\rho^{\otimes 3}\Pi^\dagger_3]$, which means we only need to construct a quantum circuit which is block-encoded by the projector $\Pi_3$. Specifically, we will show $(\bra{0}\otimes I_3)U_{3\to1}(\ket{0}\otimes I_3)=\Pi_3$. Notice that
\begin{equation}\label{eq:u3-1-derivation}
\begin{aligned}
    (\bra{0}\otimes I_3)U_{3\to1}(\ket{0}\otimes I_3)&=\bra{0}R_y(\gamma)\ket{0}\bra{0}R_y(\beta)\ket{0}\bra{0}R_y(\alpha)\ket{0}I_3\\
    &\quad+\bra{0}R_y(\gamma)\ket{0}\bra{0}R_y(\beta)\ket{1}\bra{1}R_y(\alpha)\ket{0}\mathbf{P}_{3}((123))\\
    &\quad+\bra{0}R_y(\gamma)\ket{1}\bra{1}R_y(\beta)\ket{0}\bra{0}R_y(\alpha)\ket{0}\mathbf{P}_{3}((132))\\
    &\quad+\bra{0}R_y(\gamma)\ket{1}\bra{1}R_y(\beta)\ket{1}\bra{1}R_y(\alpha)\ket{0}\mathbf{P}_{3}((132))\mathbf{P}_{3}((123)),
\end{aligned}
\end{equation}
and permutation operators matrix of the symmetry group $\cS_3$ satisfy $\mathbf{P}_{3}((12))+\mathbf{P}_{3}((13))+\mathbf{P}_3((23))=\mathbf{P}_{3}((1))+\mathbf{P}_3((123))+\mathbf{P}_3((132))$ for the qubit system, which implies that the projector on the symmetric subspace $\cH^{\otimes 3}_A$ can be rewritten as $\Pi_3=\frac{1}{3}\left[\mathbf{P}_3((1))+\mathbf{P}_3((123))+\mathbf{P}_3((132))\right]$, and $\mathbf{P}_3((123))\mathbf{P}_3((132))=\mathbf{P}_3((1))$. It is straightforward to check that one can set $\alpha=\gamma=-\arctan(\sqrt{2})$ and $\beta=\arccos(-1/3)$ such that Eq.~\eqref{eq:u3-1-derivation} equals to $\Pi_3$, which completes this proof. The specific circuit is shown in Fig.~\ref{fig:protocol_3_1}.
\end{proof}

\begin{figure}[H]
    \centering
    \includegraphics[width=0.55\textwidth]{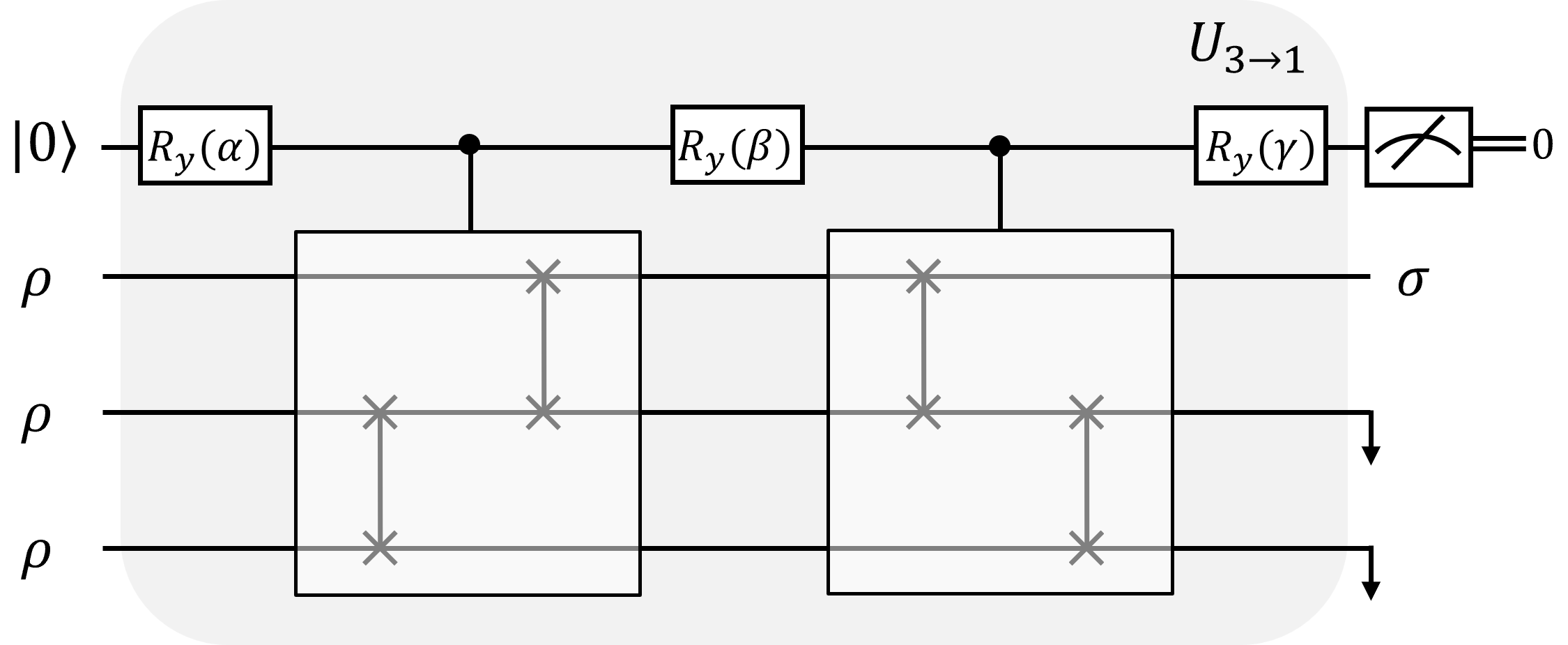}
    \caption{The simplified purification circuit for given three input noisy qubit states. We set the rotation parameters $\alpha=\gamma=-\arctan(\sqrt{2})$ and $\beta=\arccos(-1/3)$, and apply measurement on  the auxiliary system. Then, we obtain the purified state $\sigma$ if the measurement outcome is zero.}
    \label{fig:protocol_3_1}
\end{figure}

The purification circuit illustrated in Fig.~\ref{fig:protocol_3_1} demonstrates a practical approach to purify noise states using a minimalistic design. 
The circuit employs just three $R_y$ rotation gates acting on a single ancilla qubit, showcasing an efficient use of quantum resources. The key to this efficiency lies in the selective use of permutation operators, specifically $\mathbf{P}_3((123))$ and $\mathbf{P}_3((132))$. This selection is informed by the property that the symmetric projector $\Pi_3$ for three qubits can be expressed as a linear combination: $\Pi_3=\frac{1}{3}\left[\mathbf{P}_3((1))+\mathbf{P}_3((123))+\mathbf{P}_3((132))\right]$, where $\mathbf{P}_3((123))\mathbf{P}_3((132))=\mathbf{P}_3((1))$. Furthermore, it is noteworthy that, utilizing this protocol, one can quantify the purification performance with respect to the noisy parameter $\delta$. In specific, with any input noisy state $\rho^{\otimes 3}$, this circuit allows us to produce an output state $\sigma$ that exhibits high fidelity with the ideal pure state $\ket{\psi}$. The output state $\sigma$ can be rewritten as the following form: 

\begin{equation}
\begin{aligned} 
\sigma=\frac{\frac{1}{3}\rho+\frac{2}{3}\rho^3}{\frac{1}{3}+\frac{2}{3}\tr[\rho^3]}=(1-\delta^\prime)\ketbra{\psi}{\psi}+\delta^\prime\frac{I}{2},\quad \delta^\prime:=\frac{2\delta+\delta^3}{6p(\delta)}, 
\end{aligned} 
\end{equation}
where $p(\delta):=1-\delta + \frac{1}{2}\delta^2$ represents the probability of obtaining the measurement outcome zero, i.e., the purification probability. This purification circuit can be treated as a gadget and can be utilized to construct a recursive purification protocol. Further analysis of stream purification can be done via the technique in Ref.~\cite{childs2023streaming}. 

Similarly, when considering four copies of a noisy state are involved, a practical quantum circuit is proposed and shown in Fig.~\ref{fig:U_4_1}. In this case, the selection of rotation angles and specific permutation operators are relatively difficult. Thus, we treat this circuit as a PQC, where the rotation angle $\alpha_j$, for $j=0,\cdots,5$ are trainable parameters and the objective is to train a circuit $U_{4\to 1}$ such that it can approximate the projector $\Pi_4$, i.e., $(\bra{0}\otimes I_4)U_{4\to1}(\ket{0}\otimes I_4)=\Pi_4$. Obviously, the choice of purification operators is a non-trivial issue. We observe that two equivalent classes, denoted by $[4]:=\{(1234),(1243),(1324),(1342),(1423),(1432)\}$ and $[2,2]:=\{(12)(34),(13)(24),(14)(23)\}$, are sufficient to represent the projector $\Pi_4$ in terms of the following property:
\begin{equation}
    \begin{aligned}
        \Pi_4 = \frac{3\mathbf{P}_{4}((1)) + 2\mathbf{P}_{4}([4]) - \mathbf{P}_{4}([2,2])}{12},
    \end{aligned}
\end{equation}
where $\mathbf{P}_{4}([4])$ and $\mathbf{P}_{4}([2,2])$ denote the sum of permutation operator over equivalent classes $[4]$ and $[2,2]$, respectively. Thus, by selecting rotation angles using analytical or PQC strategies, and carefully designing the permutation operators, one might present a quantum circuit for symmetric projection in general cases, although there exists a trade-off between the usage of ancilla qubits and the complexity of choosing permutation operators.

\begin{figure}[H]
    \centering
    \includegraphics[width=1\linewidth]{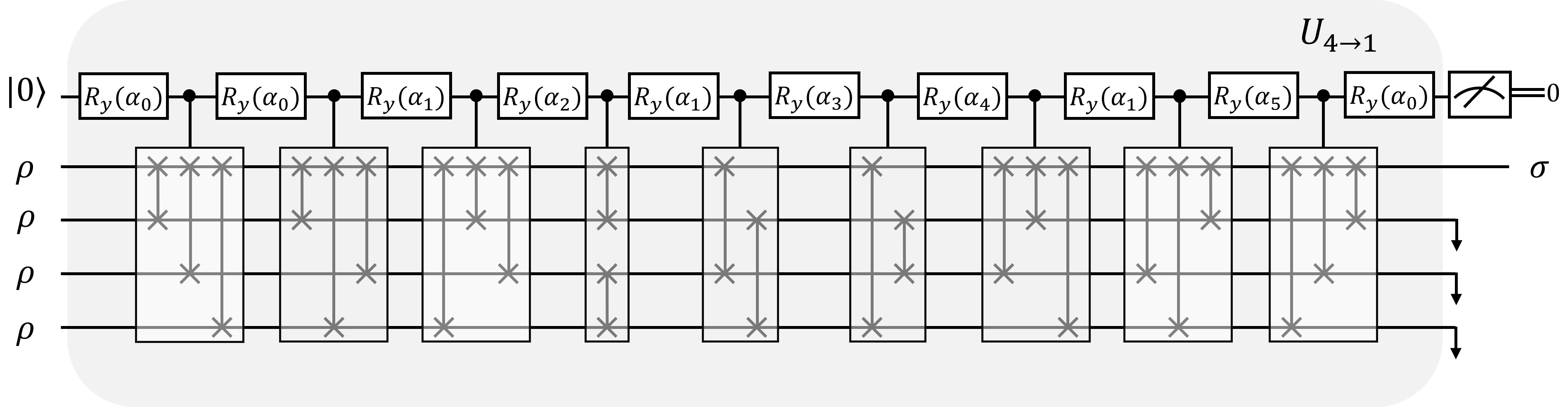}
    \caption{The simplified purification circuit for given four input noisy qubit states. We set the rotation 
approximately parameters $\alpha_0= 1.8447 + \pi$, $\alpha_1= 3.6535$, $\alpha_2= 4.3632$, $\alpha_3= \pi$, $\alpha_4= 2.4319$ and $\alpha_5= 1.8447$. Notice that the analytical solutions for the parameters $\{\alpha_j\}_{j=0}^5$ are not given. Therefore, the specific numerical value may vary slightly across different computers.}
    \label{fig:U_4_1}
\end{figure}

%%%%%%%%%%%%%%%%%%%%%%%%%%%%%%%%%%%%%%%%%%%%%%%%%%%%%%%%%%%%%%%%%%%%%%%%%%%%
%%%%%%%%%%%%%%%%%%%%%%%%%%%%%%%%%%%%%%%%%%%%%%%%%%%%%%%%%%%%%%%%%%%%%%%%%%%%
\section{Sample complexity estimation and recursive extension}\label{sec:experiments}
In this section, we evaluate the performance of the probabilistic purification protocols that is the symmetric projection operation. First, we present an estimation algorithm for the sample complexity of the optimal protocol $\cE^\ast_{A^n\to A}$. Then, we propose a general recursive purification protocol and numerically demonstrate the relationship between maximal fidelity and recursive depth. 

%,dephasing noise~\cite{arqand2020quantum}, 

%%%%%%%%%%%%%%%%%%%%%%%%%%%%%%%%%%%%%%%%%%%%%%%%%%%%%%%%%%%%%%%%%%%%%%%%%%%%
% \subsection{Sample complexity estimation}\label{sec:Sample complexity estimation}
As the probabilistic purification protocol $\cE^\ast_{A^n\to A}$, the number of depolarized states consumed in this protocol is of interest, also known as the \textit{sample complexity}. To estimate the sample complexity, we introduce a classical algorithm based on the recursive forms of maximal fidelity $f_n$ and probability $p_n$ proposed in Lemma~\ref{lem:recursion_expression}. The algorithm is designed to handle an input error parameter $\delta$, system dimension $d$, the desired output fidelity $f_{goal}$, and return the number of input states $n$ during this protocol and the expected number of noisy states $\frac{n}{p_n}$. By utilizing Alg.~\ref{Alg:complexity_estimation}, one can determine the expected numbers of noisy states consumed in the probabilistic protocol $\cE^\ast_{A^n\to A}$ for any target fidelity $f_{goal}\in(0,1)$ and system dimension $d\geq 2$. The scaling of the sample complexity is presented in Table~\ref{tab:sample_complexity}.

\begin{algorithm}[H]\label{Alg:complexity_estimation}
  \SetAlgoLined
  \SetKwData{Left}{left}\SetKwData{This}{this}\SetKwData{Up}{up}
  \SetKwFunction{Union}{Union}\SetKwFunction{FindCompress}{FindCompress}
  \SetKwInOut{Input}{input}\SetKwInOut{Output}{output}
  
  \Input{Goal fidelity $f_{goal}$, error parameter $\delta$ for depolarizing noise and dimension $d$ of Hilbert spaces $\cH_A$}
  
  \Output{The number of input noisy states and the expected number of noisy states consumed during this protocol}
  $n \leftarrow 0$, $p_0 \leftarrow 1$
  and $\Lambda := \operatorname{Diag}(\lambda_0\,\cdots,\lambda_{d-1})$   \\$\lambda_0 \leftarrow 1-\frac{d-1}{d}\delta$, $\lambda_{j \neq 0} \leftarrow \frac{\delta}{d}$\\
  \Repeat{$f_n \geq f_{goal}$}
  {Number of input noisy states $n \leftarrow n+1$ \\ $n_{th}$ success probability $p_n \leftarrow \frac{1}{n}\sum_{j=1}^n p_{n-j} \tr\left[ \Lambda^j\right]$ \\ $n_{th}$  output fidelity $f_n \leftarrow \frac{1}{n}\sum_{j = 1}^np_{n-j} \lambda_0^j/p_n$}
  \KwRet $n$, $\frac{n}{p_n}$
  \caption{Estimation of Sample Complexity}
  \label{alg:pqc-based}
\end{algorithm}
\begin{table}[H]
\centering
\setlength{\tabcolsep}{1em}
\resizebox{0.95\textwidth}{!}{
\begin{tabular}{l|cccccccc}
\toprule
\midrule
\addlinespace
$f_{goal}$ 
& $0.9285$ 
& $0.9682$
& $0.9801$
& $0.9842$
& $0.9880$
& $0.9894$
& $0.9900$
\\
\midrule
\addlinespace
${n}$
& $3$
& $8$ 
& $14$ 
& $18$ 
& $23$ 
& $26$ 
& $28$
\\
\midrule
\addlinespace
$\frac{n}{p_n}$
& $8$
& $5.2 \times 10^1$ 
& $3.27 \times 10^2$
& $1.01 \times 10^3$
& $3.89  \times 10^3$
& $8.55 \times 10^3$
& $1.43 \times 10^4$
\\
\midrule
\bottomrule
\end{tabular}
}
\caption{Estimation of sample complexity. $f_{goal}$ denotes the target fidelity, $n$ is the number of the input noisy state, and $\frac{n}{p_n}$ is the expected numbers of depolarized states $\rho$ with error parameter $\delta=0.3$ and dimension $d=3$.}
\label{tab:sample_complexity}
\end{table}

%%%%%%%%%%%%%%%%%%%%%%%%%%%%%%%%%%%%%%%%%%%%%%%%%%%%%%%%%%%%%%%%%%%%%%%%%%%%
% \subsection{Recursive purification}\label{sec:recursive purification}
It is natural to consider the protocol $\cE^\ast_{A^n\to A}$ as foundational element to construct a recursive purification procedure, which aims at enhancing fidelity as the number of iterations increases. Specifically, we propose a general recursive purification method as follows:

\begin{algorithm}[H]
  \SetAlgoLined
  \SetKwData{Left}{left}\SetKwData{This}{this}\SetKwData{Up}{up}
  \SetKwFunction{Union}{Union}\SetKwFunction{FindCompress}{FindCompress}
  \SetKwInOut{Input}{input}\SetKwInOut{Output}{output}
  
  \Input{The number of the noise qudit $\rho$ consumed per recursion $n$, and the recursive depth $m$}
  \Output{Purified state $\sigma_{n,m}$}
  
  \eIf{m=0}{
   return $\sigma_{n, 0} \leftarrow \rho$\\
  }{
  $\sigma_{n,m-1} \leftarrow $GR-PURIFY$(n, m-1)$\\
  \Repeat{a = 0}
  {Apply the purification circuit mentioned in Proposition~\ref{prop:circuit of protocol} to $\sigma_{n,m-1}$ and denote the measurement outcome by $a$.}
  \KwRet Purified state $\sigma_{n,m}$}
  \caption{General Recursive Purification, named GR-PURIFY$(n, m)$}
  \label{alg:recursive_pqc-based}
\end{algorithm}

It is important to note that, in~\cite{childs2023streaming}, the authors introduce a recursive purification protocol based on the swap test procedure, which corresponds to the case of $n=2$ in our protocol. However, maintaining long coherence for purified states can pose a significant challenge for quantum memory as the depth of the recursive protocol increases. Compared to the recursive purification procedure based on the swap test, It is evident from Fig.~\ref{fig:recursive_depth_fidelity} that the recursive purification protocol proposed in Alg.~\ref{alg:recursive_pqc-based}  requires a lower recursive depth to achieve the same target fidelity.

\begin{figure}[H]
\centering 
\includegraphics[width=0.53\textwidth]{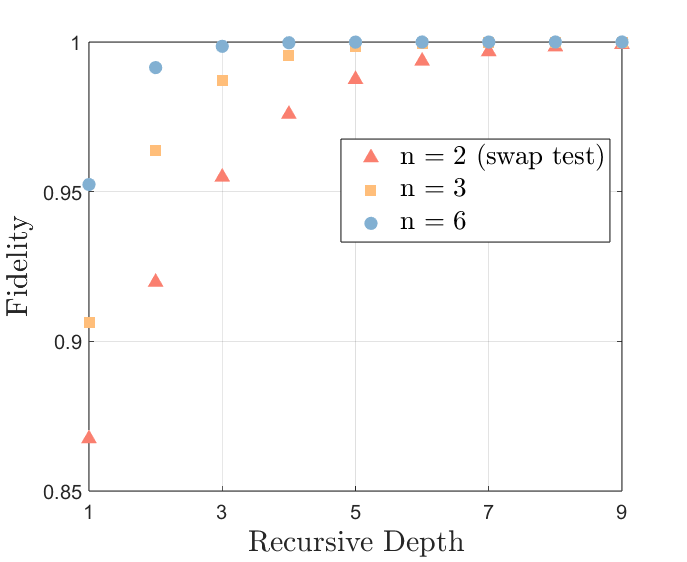}
\caption{Trend of fidelity between output state $\sigma_{n,m}$ presented in Alg.~\ref{alg:recursive_pqc-based} and the corresponding ideal state as recursive depth increases. The error parameter of depolarizing noise is $\delta = 0.3$ and the dimension of noisy states is $d=3$.}
\label{fig:recursive_depth_fidelity}
\end{figure}

%%%%%%%%%%%%%%%%%%%%%%%%%%%%%%%%%%%%%%%%%%%%%%%%%%%%%%%%%%%%%%%%%%%%%%%%%%%%
%%%%%%%%%%%%%%%%%%%%%%%%%%%%%%%%%%%%%%%%%%%%%%%%%%%%%%%%%%%%%%%%%%%%%%%%%%%%
\section{Conclusions and Discussions} 
In this paper, we focus on the probabilistic purification task. Firstly, we have formulated a semidefinite programming framework to compute the optimal probabilistic purification protocols for different settings.
This framework gives us the flexibility to explore the optimal protocols in practical scenarios where the parameters: the fidelity, the success probability, copies of noisy states, and noise models, may have different priorities.

Specifically, for the depolarizing noise, we have demonstrated that a general probabilistic purification protocol can be derived via our SDP framework, which was originally proposed by Barenco et al.~\cite{barenco1997stabilization} aiming to project the $n$ copies of noisy state into its symmetric subspace. 
Notably, we have proved its optimality for the case of $n=2$ and any dimension $d$, which broadened our knowledge in the state purification task that has been explored in previous works~\cite{cirac1999optimal,fu2016quantum}. 
Our proof routine reduced the original problem to the investigation of the positivity of the operator $f_nR_n^{T_{A^n}}-Q_n^{T_{A^n}}$. It could pave the way for proving the optimality for the general case. 
We have also verified the optimality for a range of number of state copies $n$ and dimensions $d$ through numerical experiments shown in Table.~\ref{tab:comparison_n_d}. 
Therefore, we conjectured that the symmetric projection protocol $\cE^\ast_{A^n\to A}$ is optimal for general $n$ and $d$. 
In addition, explicit forms of $f_n$ and $p_n$ have been established to directly characterize the trade-off between fidelity and probability. 
Other noise models have been investigated through numerical experiments 
% \honghao{numerically?}
, which manifests the flexibility and potential of our purification frameworks.

Secondly, we have investigated the practicability and efficiency of the circuit implementation of the conjectured optimal protocol using the block encoding and parameterized quantum circuits techniques. 
Specifically, we have proposed explicit circuits of qubit systems for the cases of $n=3$ and $n=4$, respectively. 
It is worth noting that these circuits involve only one ancilla qubit, dramatically reducing the width of circuits compared with previous works~\cite{barenco1997stabilization,laborde2024quantum}. 
Thus, our analysis method used in the circuits described above plays a crucial role in practical settings.

Finally, a classical algorithm has been proposed to estimate the sample complexity of the protocol $\cE_{A^n\to A}^\ast$, which allows us to determine the expected numbers
of noisy states consumed in the probabilistic protocol for a given target fidelity. 
We have further observed that this protocol can also be used as a gadget in a recursive purification protocol. 
Our recursive protocol involves lower recursive depth and smaller quantum memory than the protocol in~\cite{childs2023streaming} for a given target fidelity. In general, the estimation algorithm and the recursive protocol pave the way for the application of state purification protocols in the real world.

For future work, a theoretical proof of the optimality of our purification protocol is desirable for the general case of $n>2$ and arbitrary $d$, although we have given numerical evidence. 
In addition, it will be meaningful to explore how to simplify the purification circuit to reduce the number of ancilla qubits for general $n$ and $d$, thus enhancing practical applicability. 
Another interesting direction is to analyze the optimal probabilistic purification protocol for various noise channels through SDP characterizations and by exploring the properties of different noises. 
It will also be interesting to further extend or apply the techniques of this work to quantum error mitigation (e.g. virtual state purification~\cite{koczor2021exponential, huggins2021virtual, liu2024virtual}) and to other purification tasks. 
In particular, the process similar to state purification, known as entanglement distillation~\cite{bennett1996purification, deutsch1996quantum, pan2003experimental, kalb2017entanglement,Devetak2005,dur2007entanglement,fang2019non}, plays a pivotal role in the realm of distributed quantum information processing and is of particular significance in the context of quantum networks~\cite{dur1999quantum, kimble2008quantum}.
Alongside the practical aspects of entanglement distillation, researchers have also delved into the fundamental limitations and possibilities of purifying quantum resources from noisy quantum states. These studies have been carried out within the framework of various quantum resource theories~\cite{fang2020no,regula2021fundamental,hastings2018distillation,wang2020efficiently,fang2018probabilistic,wang2019resource,chitambar2019quantum}. Thus, it will be promising if our SDP framework and analysis methods are utilized in the above significant research areas,
for example, the distillation of entanglement and other quantum resources.

\section*{Acknowledgement}
We would like to thank Zanqiu Shen, Guangxi Li, Zhan Yu, Yingjian Liu, Yu Gan, and Keming He, for their helpful discussions and comments.  This work was supported by the National Key R\&D Program of China (Grant No.~2024YFE0102500), the Guangdong Provincial Quantum Science Strategic Initiative (Grant No.~GDZX2303007), the Guangdong Provincial Key Lab of Integrated Communication, Sensing and Computation for Ubiquitous Internet of Things (Grant No.~2023B1212010007), the Quantum Science Center of Guangdong-Hong Kong-Macao Greater Bay Area, and the Education Bureau of Guangzhou Municipality.
%%%%%%%%%%%%%%%%%%%%%%%%%%%%%%%%%%%%%%%%%%%%%%%%%%%%%%%%%%%%%%%%%%%%%%%%%%%
%%%%%%%%%%%%%%%%%%%%%%%%%%%%%%%%%%%%%%%%%%%%%%%%%%%%%%%%%%%%%%%%%%%%%%%%%
% Bibliography

\bibliographystyle{alpha}
\bibliography{references}

\newcommand{\etalchar}[1]{$^{#1}$}
\begin{thebibliography}{BCWDW01}

\bibitem[AAB{\etalchar{+}}19]{arute2019quantum}
Frank Arute, Kunal Arya, Ryan Babbush, Dave Bacon, Joseph~C Bardin, Rami Barends, Rupak Biswas, Sergio Boixo, Fernando~GSL Brandao, David~A Buell, et~al.
\newblock Quantum supremacy using a programmable superconducting processor.
\newblock {\em Nature}, 574(7779):505--510, 2019.

\bibitem[ARS88]{alicki1988symmetry}
Robert Alicki, S{\l{}}awomir Rudnicki, and S{\l{}}awomir Sadowski.
\newblock Symmetry properties of product states for the system of {N} $n$-level atoms.
\newblock {\em Journal of Mathematical Physics}, 29(5):1158--1162, 1988.

\bibitem[BBD{\etalchar{+}}97]{barenco1997stabilization}
Adriano Barenco, Andre Berthiaume, David Deutsch, Artur Ekert, Richard Jozsa, and Chiara Macchiavello.
\newblock Stabilization of quantum computations by symmetrization.
\newblock {\em SIAM Journal on Computing}, 26(5):1541--1557, 1997.

\bibitem[BBP{\etalchar{+}}96]{bennett1996purification}
Charles~H Bennett, Gilles Brassard, Sandu Popescu, Benjamin Schumacher, John~A Smolin, and William~K Wootters.
\newblock Purification of noisy entanglement and faithful teleportation via noisy channels.
\newblock {\em Physical Review Letters}, 76(5):722, 1996.

\bibitem[BCWDW01]{buhrman2001quantum}
Harry Buhrman, Richard Cleve, John Watrous, and Ronald De~Wolf.
\newblock Quantum fingerprinting.
\newblock {\em Physical Review Letters}, 87(16):167902, 2001.

\bibitem[BLSF19]{benedetti2019parameterized}
Marcello Benedetti, Erika Lloyd, Stefan Sack, and Mattia Fiorentini.
\newblock Parameterized quantum circuits as machine learning models.
\newblock {\em Quantum Science and Technology}, 4(4):043001, 2019.

\bibitem[BV04]{boyd2004convex}
Stephen Boyd and Lieven Vandenberghe.
\newblock {\em Convex optimization}.
\newblock Cambridge university press, 2004.

\bibitem[CEM99]{cirac1999optimal}
J~Ignacio Cirac, AK~Ekert, and Chiara Macchiavello.
\newblock Optimal purification of single qubits.
\newblock {\em Physical Review Letters}, 82(21):4344, 1999.

\bibitem[CFL{\etalchar{+}}23]{childs2023streaming}
Andrew~M Childs, Honghao Fu, Debbie Leung, Zhi Li, Maris Ozols, and Vedang Vyas.
\newblock Streaming quantum state purification.
\newblock {\em arXiv preprint arXiv:2309.16387}, 2023.

\bibitem[CG19]{chitambar2019quantum}
Eric Chitambar and Gilad Gour.
\newblock Quantum resource theories.
\newblock {\em Reviews of modern physics}, 91(2):025001, 2019.

\bibitem[Cho75]{choi1975completely}
Man-Duen Choi.
\newblock Completely positive linear maps on complex matrices.
\newblock {\em Linear algebra and its applications}, 10(3):285--290, 1975.

\bibitem[CLVBY24]{camps2024explicit}
Daan Camps, Lin Lin, Roel Van~Beeumen, and Chao Yang.
\newblock Explicit quantum circuits for block encodings of certain sparse matrices.
\newblock {\em SIAM Journal on Matrix Analysis and Applications}, 45(1):801--827, 2024.

\bibitem[CW12]{childs2012hamiltonian}
Andrew~M Childs and Nathan Wiebe.
\newblock Hamiltonian simulation using linear combinations of unitary operations.
\newblock {\em arXiv preprint arXiv:1202.5822}, 2012.

\bibitem[DB07]{dur2007entanglement}
Wolfgang D{\"u}r and Hans~J Briegel.
\newblock Entanglement purification and quantum error correction.
\newblock {\em Reports on Progress in Physics}, 70(8):1381, 2007.

\bibitem[DBCZ99]{dur1999quantum}
Wolfgang D{\"u}r, H-J Briegel, Juan~Ignacio Cirac, and Peter Zoller.
\newblock Quantum repeaters based on entanglement purification.
\newblock {\em Physical Review A}, 59(1):169, 1999.

\bibitem[DEJ{\etalchar{+}}96]{deutsch1996quantum}
David Deutsch, Artur Ekert, Richard Jozsa, Chiara Macchiavello, Sandu Popescu, and Anna Sanpera.
\newblock Quantum privacy amplification and the security of quantum cryptography over noisy channels.
\newblock {\em Physical Review Letters}, 77(13):2818, 1996.

\bibitem[DW05]{Devetak2005}
Igor Devetak and Andreas Winter.
\newblock Distillation of secret key and entanglement from quantum states.
\newblock {\em Proceedings of the Royal Society A: Mathematical, Physical and engineering sciences}, 461(2053):207--235, 2005.

\bibitem[EW01]{eggeling2001separability}
Tilo Eggeling and Reinhard~F Werner.
\newblock Separability properties of tripartite states with {U $\ox$ U $\ox$ U} symmetry.
\newblock {\em Physical Review A}, 63(4):042111, 2001.

\bibitem[EWP{\etalchar{+}}19]{erhard2019characterizing}
Alexander Erhard, Joel~J Wallman, Lukas Postler, Michael Meth, Roman Stricker, Esteban~A Martinez, Philipp Schindler, Thomas Monz, Joseph Emerson, and Rainer Blatt.
\newblock Characterizing large-scale quantum computers via cycle benchmarking.
\newblock {\em Nature communications}, 10(1):5347, 2019.

\bibitem[Fiu04]{fiuravsek2004optimal}
Jarom{\'\i}r Fiur{\'a}{\v{s}}ek.
\newblock Optimal probabilistic cloning and purification of quantum states.
\newblock {\em Physical Review A—Atomic, Molecular, and Optical Physics}, 70(3):032308, 2004.

\bibitem[FL20]{fang2020no}
Kun Fang and Zi-Wen Liu.
\newblock No-go theorems for quantum resource purification.
\newblock {\em Physical Review Letters}, 125(6):060405, 2020.

\bibitem[Fu16]{fu2016quantum}
Honghao Fu.
\newblock Quantum state purification.
\newblock Master's thesis, University of Waterloo, 2016.

\bibitem[FWL{\etalchar{+}}18]{fang2018probabilistic}
Kun Fang, Xin Wang, Ludovico Lami, Bartosz Regula, and Gerardo Adesso.
\newblock Probabilistic distillation of quantum coherence.
\newblock {\em Physical Review Letters}, 121(7):070404, 2018.

\bibitem[FWTD19]{fang2019non}
Kun Fang, Xin Wang, Marco Tomamichel, and Runyao Duan.
\newblock Non-asymptotic entanglement distillation.
\newblock {\em IEEE Transactions on Information Theory}, 65(10):6454--6465, 2019.

\bibitem[GSLW19]{gilyen2019quantum}
Andr{\'a}s Gily{\'e}n, Yuan Su, Guang~Hao Low, and Nathan Wiebe.
\newblock Quantum singular value transformation and beyond: exponential improvements for quantum matrix arithmetics.
\newblock In {\em Proceedings of the 51st Annual ACM SIGACT Symposium on Theory of Computing}, pages 193--204, 2019.

\bibitem[Har13]{harrow2013church}
Aram~W Harrow.
\newblock The church of the symmetric subspace.
\newblock {\em arXiv preprint arXiv:1308.6595}, 2013.

\bibitem[HH18]{hastings2018distillation}
Matthew~B Hastings and Jeongwan Haah.
\newblock Distillation with sublogarithmic overhead.
\newblock {\em Physical Review Letters}, 120(5):050504, 2018.

\bibitem[HMO{\etalchar{+}}21]{huggins2021virtual}
William~J Huggins, Sam McArdle, Thomas~E O’Brien, Joonho Lee, Nicholas~C Rubin, Sergio Boixo, K~Birgitta Whaley, Ryan Babbush, and Jarrod~R McClean.
\newblock Virtual distillation for quantum error mitigation.
\newblock {\em Physical Review X}, 11(4):041036, 2021.

\bibitem[HSFL14]{hou2014experimental}
Shi-Yao Hou, Yu-Bo Sheng, Guan-Ru Feng, and Gui-Lu Long.
\newblock Experimental optimal single qubit purification in an nmr quantum information processor.
\newblock {\em Scientific reports}, 4(1):6857, 2014.

\bibitem[Jam72]{jamiolkowski1972linear}
Andrzej Jamio{\l}kowski.
\newblock Linear transformations which preserve trace and positive semidefiniteness of operators.
\newblock {\em Reports on mathematical physics}, 3(4):275--278, 1972.

\bibitem[Kim08]{kimble2008quantum}
H~Jeff Kimble.
\newblock The quantum internet.
\newblock {\em Nature}, 453(7198):1023--1030, 2008.

\bibitem[Koc21]{koczor2021exponential}
B{\'a}lint Koczor.
\newblock Exponential error suppression for near-term quantum devices.
\newblock {\em Physical Review X}, 11(3):031057, 2021.

\bibitem[KRH{\etalchar{+}}17]{kalb2017entanglement}
Norbert Kalb, Andreas~A Reiserer, Peter~C Humphreys, Jacob~JW Bakermans, Sten~J Kamerling, Naomi~H Nickerson, Simon~C Benjamin, Daniel~J Twitchen, Matthew Markham, and Ronald Hanson.
\newblock Entanglement distillation between solid-state quantum network nodes.
\newblock {\em Science}, 356(6341):928--932, 2017.

\bibitem[KW01]{keyl2001rate}
Michael Keyl and Reinhard~F Werner.
\newblock The rate of optimal purification procedures.
\newblock In {\em Annales Henri Poincare}, volume~2, pages 1--26. Springer, 2001.

\bibitem[KW20]{khatri2020principles}
Sumeet Khatri and Mark~M Wilde.
\newblock Principles of quantum communication theory: A modern approach.
\newblock {\em arXiv preprint arXiv:2011.04672}, 2020.

\bibitem[LC19]{low2019hamiltonian}
Guang~Hao Low and Isaac~L Chuang.
\newblock Hamiltonian simulation by qubitization.
\newblock {\em Quantum}, 3:163, 2019.

\bibitem[LKWI23]{loaiza2023reducing}
Ignacio Loaiza, Alireza~Marefat Khah, Nathan Wiebe, and Artur~F Izmaylov.
\newblock Reducing molecular electronic hamiltonian simulation cost for linear combination of unitaries approaches.
\newblock {\em Quantum Science and Technology}, 8(3):035019, 2023.

\bibitem[LRW24]{laborde2024quantum}
Margarite~L LaBorde, Soorya Rethinasamy, and Mark~M Wilde.
\newblock Quantum algorithms for realizing symmetric, asymmetric, and antisymmetric projectors.
\newblock {\em arXiv preprint arXiv:2407.17563}, 2024.

\bibitem[LZFC24]{liu2024virtual}
Zhenhuan Liu, Xingjian Zhang, Yue-Yang Fei, and Zhenyu Cai.
\newblock Virtual channel purification.
\newblock {\em arXiv preprint arXiv:2402.07866}, 2024.

\bibitem[NC01]{nielsen2001quantum}
Michael~A Nielsen and Isaac~L Chuang.
\newblock {\em Quantum computation and quantum information}, volume~2.
\newblock Cambridge university press Cambridge, 2001.

\bibitem[PGU{\etalchar{+}}03]{pan2003experimental}
Jian-Wei Pan, Sara Gasparoni, Rupert Ursin, Gregor Weihs, and Anton Zeilinger.
\newblock Experimental entanglement purification of arbitrary unknown states.
\newblock {\em Nature}, 423(6938):417--422, 2003.

\bibitem[Ple10]{pletsch2010recursion}
Bill Pletsch.
\newblock A recursion of the p{\'o}lya polynomial for the symmetric group.
\newblock {\em Mathematics and Computers in Simulation}, 80(6):1212--1220, 2010.

\bibitem[Pre18]{preskill2018quantum}
John Preskill.
\newblock Quantum computing in the nisq era and beyond.
\newblock {\em Quantum}, 2:79, 2018.

\bibitem[RMC{\etalchar{+}}04]{ricci2004experimental}
Marco Ricci, F~De Martini, NJ~Cerf, R~Filip, J~Fiur{\'a}{\v{s}}ek, and Chiara Macchiavello.
\newblock Experimental purification of single qubits.
\newblock {\em Physical Review Letters}, 93(17):170501, 2004.

\bibitem[RT21]{regula2021fundamental}
Bartosz Regula and Ryuji Takagi.
\newblock Fundamental limitations on distillation of quantum channel resources.
\newblock {\em Nature Communications}, 12(1):4411, 2021.

\bibitem[Sho95]{shor1995scheme}
Peter~W Shor.
\newblock Scheme for reducing decoherence in quantum computer memory.
\newblock {\em Physical Review A}, 52(4):R2493, 1995.

\bibitem[Ste96]{steane1996error}
Andrew~M Steane.
\newblock Error correcting codes in quantum theory.
\newblock {\em Physical Review Letters}, 77(5):793, 1996.

\bibitem[UNP{\etalchar{+}}21]{urbanek2021mitigating}
Miroslav Urbanek, Benjamin Nachman, Vincent~R Pascuzzi, Andre He, Christian~W Bauer, and Wibe~A de~Jong.
\newblock Mitigating depolarizing noise on quantum computers with noise-estimation circuits.
\newblock {\em Physical Review Letters}, 127(27):270502, 2021.

\bibitem[WW19]{wang2019resource}
Xin Wang and Mark~M Wilde.
\newblock Resource theory of asymmetric distinguishability for quantum channels.
\newblock {\em Physical Review Research}, 1(3):033169, 2019.

\bibitem[WWS20]{wang2020efficiently}
Xin Wang, Mark~M Wilde, and Yuan Su.
\newblock Efficiently computable bounds for magic state distillation.
\newblock {\em Physical Review Letters}, 124(9):090505, 2020.

\bibitem[YKLO24]{yang2024quantum}
Bo~Yang, Elham Kashefi, Dominik Leichtle, and Harold Ollivier.
\newblock Quantum error suppression with subgroup stabilisation projectors.
\newblock {\em arXiv preprint arXiv:2404.09973}, 2024.

\end{thebibliography}

\newpage
%%%%%%%%% SUPPLEMENTAL MATERIAL %%%%%%%%%

\appendix
\setcounter{subsection}{0}
\setcounter{table}{0}
\setcounter{figure}{0}

\vspace{3cm}
% \onecolumngrid
% \vspace{2cm}

\begin{center}
\Large{\textbf{Appendix for} \\ \textbf{
Protocols and Trade-Offs of Quantum State Purification}}
\end{center}

\renewcommand{\theproposition}{S\arabic{proposition}}
\setcounter{proposition}{0}

% \tableofcontents

%%%%%%%%%%%%%%%%%%%%%%%%%%%%%%%%%%%%%%%%%%%%%%%%%%%%%%%%%%%%%%%%%%%%%%%%%%%%%%%%%
%%%%%%%%%%%%%%%%%%%%%%%%%%%%%%%%%%%%%%%%%%%%%%%%%%%%%%%%%%%%%%%%%%%%%%%%%%%%%%%%%
\section{Dual SDP for purification protocols}\label{appendix:dual_sdp}
For a given purification success probability $p$ and number of input state copies $n$, The primal SDP for calculating the optimal purification fidelity $F_{\cN}(n,p)$ is written as follows:
\begin{equation}
\begin{aligned}
    F_{\cN}(n,p)=\max\;&\tr[J^{\cE}_{{A^n A}}Q_n^{T_{A^n}}]/p\\
    {\rm s.t.}\;& \tr[J^{\cE}_{{A^nA}}R_n^{T_{A^n}}]= p,\\
    &J^{\cE}_{A^n A}\geq 0,\,\tr_{ A}[J^{\cE}_{A^n A}] \leq I_{A^n},
\end{aligned}
\end{equation}
where $T_{A^n}$ denotes partial transpose on system $A^n$ and 
\begin{equation}
    \begin{aligned}
        Q_n:=\int\cN(\psi)^{\otimes n}\otimes \psi\,d\psi, \quad  R_n:=\int\cN(\psi)^{\otimes n}\otimes I_A\,d\psi.
    \end{aligned}
\end{equation}
Now, we derive its dual SDP. Based on the primal SDP, the Lagrange function can be written as
\begin{align}
    \mathcal{L}(x, Y_{A^n}, J^{\cE}_{{A^n A}}):&= \frac{1}{p}\tr[J^{\cE}_{{A^n A}}Q_n^{T_{A^n}}] +\langle x,\tr[J^{\cE}_{{A^n A}}R_n^{T_{A^n}}] - p\rangle + \langle Y_{A^n}, \tr_{A}[J^{\cE}_{{A^n A}}]-I_{A^n}\rangle\\
    &=-px-\tr[Y_{A^n}]+\langle J^{\cE}_{{A^n A}}, \frac{1}{p}Q_n^{T_{A^n}}+R_n^{T_{A^n}}x+Y_{A^n}\otimes I_{A} \rangle,
\end{align} 
where $x$, $Y_{A^n}$ are Lagrange multipliers. Then, the Lagrange dual function can be written as
\begin{align}
     \mathcal{G}(x, &Y_{A^n}):=\sup_{J^{\cE}_{{A^n A}}\geq0}\mathcal{L}(x, Y_{A^n}, J^{\cE}_{{A^n A}}).
\end{align}
Since $J^{\cE}_{{A^n A}}\geq0$, it holds that
\begin{align}
    &\frac{1}{p}Q_n^{T_{A^n}}+R_n^{T_{A^n}}x+Y_{A^n}\otimes I_{A}\leq 0,
\end{align}
otherwise, the inner norm is unbounded. Similarly, we have $Y_{A^n}\leq 0$ since $\tr_{ A}[J^{\cE}_{A^n A}] \leq I_{A^n}$. Redefine $x$ as $x/p$ and $Y_{A^n}$ as $Y_{A^n}/p$. Then, we obtain the following dual SDP.
\begin{equation}
\begin{aligned}
    \min\;& -x-\tr[Y_{A^n}]/p\\
    {\rm s.t.}\;& Q_n^{T_{A^n}}+R_n^{T_{A^n}}x+Y_{A^n}\otimes I_{A}\leq 0\\
    &Y_{A^n}\leq 0.
\end{aligned}
\end{equation}
It is worth noting that the strong duality is held by Slater's condition.

Similarly, the primal SDP for calculating the optimal purification success probability $P_{\cN}(n,f)$ is written as follows:
\begin{equation}
\begin{aligned}
    P_{\cN}(n,f)=\max\;&\tr[J^{\cE}_{{A^n A}}R_n^{T_{A^n}}]\\
    {\rm s.t.}\;& \tr[J^{\cE}_{{A^nA}}Q_n^{T_{A^n}}]=\tr[J^{\cE}_{{A^n A}}R_n^{T_{A^n}}]f,\\
    &J^{\cE}_{A^n A}\geq 0,\,\tr_{A}[J^{\cE}_{{A^n A}}] \leq I_{A^n}.
\end{aligned}
\end{equation}
Now, we derive its dual SDP. Based on the primal SDP, the Lagrange function can be written as
\begin{align}
    \mathcal{L}(x, Y_{A^n}, J^{\cE}_{{A^n A}}):&= \tr[J^{\cE}_{{A^n A}}R_n^{T_{A^n}}] +\langle x,\tr[J^{\cE}_{{A^n A}}Q_n^{T_{A^n}}] - \tr[J^{\cE}_{{A^n A}}R_n^{T_{A^n}}]f\rangle + \langle Y_{A^n}, \tr_{A}[J^{\cE}_{{A^n A}}]-I_{A^n}\rangle\\
    &=-\tr[Y_{A^n}]+\langle J^{\cE}_{{A^n A}}, (1-xf)R_n^{T_{A^n}}+Q_n^{T_{A^n}}x+Y_{A^n}\otimes I_{A} \rangle,
\end{align} 
where $x$, $Y_{A^n}$ are Lagrange multipliers. Then, the Lagrange dual function can be written as
\begin{align}
     \mathcal{G}(x, &Y_{A^n}):=\sup_{J^{\cE}_{{A^n A}}\geq0}\mathcal{L}(x, Y_{A^n}, J^{\cE}_{{A^n A}}).
\end{align}
Since $J^{\cE}_{{A^n A}}\geq0$, it holds that
\begin{align}
    &(1-xf)R_n^{T_{A^n}}+Q_n^{T_{A^n}}x+Y_{A^n}\otimes I_{A}\leq 0,
\end{align}
otherwise, the inner norm is unbounded. Similarly, we have $Y_{A^n}\leq 0$ since $\tr_{ A}[J^{\cE}_{A^n A}] \leq I_{A^n}$. Then, we obtain the following dual SDP.
\begin{equation}
\begin{aligned}
    \min\;& -\tr[Y_{A^n}]\\
    {\rm s.t.}\;&(1-xf)R_n^{T_{A^n}}+Q_n^{T_{A^n}}x+Y_{A^n}\otimes I_{A} \leq 0\\
    &Y_{A^n}\leq 0.
\end{aligned}
\end{equation}
It is worth noting that the strong duality is held by Slater's condition.

%%%%%%%%%%%%%%%%%%%%%%%%%%%%%%%%%%%%%%%%%%%%%%%%%%%%%%%%%%%%%%%%%%%%%%%%%%%%%%%%%
%%%%%%%%%%%%%%%%%%%%%%%%%%%%%%%%%%%%%%%%%%%%%%%%%%%%%%%%%%%%%%%%%%%%%%%%%%%%%%%%%
\section{Proof of Lemma~\ref{lem:recursion_expression} and Lemma~\ref{lem:case_of_two}}
% \begin{shaded}
\begin{lemma}
Let $\delta$ be the noise parameter of the depolarizing channel and $\Pi_n$ be the projector on the symmetric subspace of $\cH^{\otimes n}_A$. Denote $p_n:=\tr[\Pi_n\Lambda^{\otimes n}]$ and $f_n:=\bra{0}\tr_{2\cdots n}[\Pi_n\Lambda^{\otimes n}\Pi_n^\dagger]\ket{0}/p_n$. Then, we have 
    \begin{equation}
        \begin{aligned}
             &p_n = \frac{1}{n}\sum_{j=1}^n p_{n-j} \tr\left[ \Lambda^j\right],\quad f_n=\frac{\sum_{j = 1}^np_{n-j} \lambda_0^j}{n p_n},
        \end{aligned}
    \end{equation}
    where $p_0=1$, $\Lambda:=\operatorname{Diag}(\lambda_0\,\cdots,\lambda_{d-1})$ denotes the eigenvalue matrix of a $d$-dimensional noisy state with $\lambda_0:=1-\frac{d-1}{d}\delta$, $\lambda_j=\frac{\delta}{d}$, for $j=1,\cdots,d-1$.
\label{lem:recursion_expression}
\end{lemma}
% \end{shaded}
\begin{proof}
Firstly, we will demonstrate the following properties:
\begin{equation}\label{eq.Recursion expression_P} 
    \tr\left[\Pi_{n} \rho^{\otimes n}\right]
    = \frac{1}{n}\sum_{j = 1}^n\tr\left[\mathbf{P}_{\{j\}} \otimes \Pi_{\{j\}^c}\rho^{\otimes n} \right],\quad \forall \rho^{\otimes n}\in\cD(\cH_A^{\otimes n})
\end{equation}
where $\mathbf{P}_{\{j\}}$ is an arbitrary cycle permutation operator of length $j$ acting on the $\{1,\cdots,j\}$ systems, $\Pi_{\{j\}^c}$ is the projector on the symmetric subspace acting on $\{j+1,\cdots,n\}$ systems, and we define $\{j\}:= \{1,\cdots,j\}$ and  $\{j\}^c:= \{1,\cdots,n\} - \{1,\cdots,j\} = \{j+1,\cdots,n\}$, a set of $n-j$ systems.
Specifically,  $K(\cS_n)$ is denoted as the number of conjugacy classes in the symmetric group $\cS_n$. The $i^{th}$ conjugacy class of $\cS_n$ can be denoted by $(l_i):= (1^{{{v_1}}},2^{{{v_2}}},\cdots, n^{{{v_n}}})$, $(l_i)\in\cS_n$, which must satisfy ${{v_1}},{{v_2}}, \cdots, {{v_n}} \in \mathbb{N^*}$, and $\sum_{j=1}^n j \cdot {{v_j}} = n$. The number of elements in the certain conjugacy class $(l_i)$ of $\cS_n$ is given by $C^{(l_i)} := \frac{n!}{1^{{v_1}}2^{{v_2}}\cdots n^{{v_n}} {{v_1}}! {{v_2}}! \cdots {{v_n}}! }$. One can find that the expression $\tr\left[\Pi_{n} \rho^{\otimes n}\right]$ can be written as follows:
    \begin{equation} \label{eq:Probability_n_copies_states}
    \begin{aligned}
    \tr\left[\Pi_{n} \rho^{\otimes n}\right]
        &=\frac{1}{n!}\sum_{i=1}^{K(\cS_n)}\sum_{i^{\prime}=1}^{C^{(l_i)}}\tr\left[\mathbf{P}_n{(c^{(l_i)}_{i^{\prime}})} \rho^{\otimes n} \right]\\
        &=\frac{1}{n!}\sum_{i=1}^{K(\cS_n)}C^{(l_i)}\tr\left[\mathbf{P}_n{(c^{(l_i)})} \rho^{\otimes n} \right]\\
        &=\tr\left[\widetilde{\Pi}_{n} \rho^{\otimes n}\right],
    \end{aligned}           
    \end{equation}
where $\widetilde{\Pi}_{n}$ denotes the projector composed of the same permutation operator in the $(l_i)$ conjugacy class, and $c^{(l_i)}_{i^{\prime}}$ and $c^{(l_i)}$ refer to the ${i^{\prime}}^{th}$ element and an arbitrary element in the conjugacy class $(l_i)$, respectively. These correspond to $\mathbf{P}_n{(c^{(l_i)}_{i^{\prime}})}$ and $\mathbf{P}_n{(c^{(l_i)})}$ which are permutation operators. Considering $n = 3$, there are $K(\cS_3) = 3$ conjugacy classes, including $(1)$, $(12)$ and $(123)$ class. Therefore, we have the following equation: 
 \begin{equation}
     \begin{aligned}
         \tr\left[\Pi_{3} \rho^{\otimes 3}\right] =& \frac{1}{3!}\tr\left[\left(\mathbf{P}_3((1))+\mathbf{P}_3((12))+\mathbf{P}_3((13))+\mathbf{P}_3((23))+\mathbf{P}_3((123))+\mathbf{P}_3((132))\right) \rho^{\otimes 3}\right]\\
         =& \frac{1}{3!}\tr\left[ \left(\mathbf{P}_3((1))+3\mathbf{P}_3((12))+2\mathbf{P}_3((123))\right)\rho^{\otimes 3}\right]=\tr\left[\widetilde{\Pi}_{3} \rho^{\otimes 3}\right].
     \end{aligned}
 \end{equation}
 Due to the P$\Acute{\rm o}$lya’s theory of counting~\cite{pletsch2010recursion}, we have 
 \begin{equation}
     \begin{aligned}
         \tr\left[\widetilde{\Pi}_{n} \rho^{\otimes n}\right]=\frac{1}{n}\sum_{j = 1}^n\tr\left[ \mathbf{P}_{\{j\}}\rho^{\otimes j} \right]\tr\left[ \Pi_{\{j\}^c}\rho^{\otimes n-j} \right].
     \end{aligned}
 \end{equation}

Secondly, we will show the recursive forms of $p_n$ and $f_n$. According to Eq.~\eqref{eq.Recursion expression_P}, the expression $p_n$ can be rewritten as following recursive form:
    \begin{equation}\label{eq.pn}
       \begin{aligned}
        p_n &=\tr\left[\Pi_{n} \Lambda^{\otimes n}\right] \\
           &= \frac{1}{n}\sum_{j = 1}^n \tr\left[\mathbf{P}_{\{j\}} \Lambda^{\otimes j}\right]\tr\left[\Pi_{\{j\}^c}\Lambda^{\otimes n-j}\right]\\
           &= \frac{1}{n}\sum_{j=1}^n p_{n-j} \tr\left[ \Lambda^j\right],
       \end{aligned}
   \end{equation}
where the third equation follows from the fact that $\tr\left[\mathbf{P}_{\{j\}}\Lambda^{\otimes j}\right]=\tr[\Lambda^j]$ and the definition of $p_n$, i.e., $p_{n-j}=\tr\left[\Pi_{\{j\}^c}\Lambda^{\otimes n-j}\right]$. We further derive the recursive form of $f_n$ as follows:
\begin{equation}\label{eq:fn}
    \begin{aligned}
        f_n:&=
        \frac{\bra{0}\tr_{2\cdots n}[\Pi_n\Lambda^{\otimes n}\Pi_n]\ket{0}}{p_n}\\
           &=\frac{1}{p_n}\tr\left[ \Pi_n  
 \cdot \ketbra{0}{0} \otimes I^{\otimes n-1} \Lambda^{\otimes n} \ketbra{0}{0} \otimes I^{\otimes n-1} \cdot \Pi_n\right]\\
           &=\frac{1}{p_n}\tr\left[ \Pi_n  \cdot \left(\lambda_0\ketbra{0}{0} \otimes \Lambda^{\otimes n-1}\right) \right],
       \end{aligned}
   \end{equation}
   where $\lambda_0$ is the first eigenvalue of any noisy state. The third equation follows from the fact that $\tr\left[ \ketbra{0}{0} \otimes I^{\otimes n-1}   
 \cdot \Pi_n \Lambda^{\otimes n} \Pi_n \cdot \ketbra{0}{0} \otimes I^{\otimes n-1} \right]=\tr\left[ \Pi_n  
 \cdot \ketbra{0}{0} \otimes I^{\otimes n-1} \Lambda^{\otimes n} \ketbra{0}{0} \otimes I^{\otimes n-1} \cdot \Pi_n\right]$ for any $\Lambda$.
   Depending on Eq.~\eqref{eq.Recursion expression_P}, the equation can be rewritten by:
   \begin{equation}
       \begin{aligned}
           f_n &=\frac{1}{n p_n} \sum_{j = 1}^n \tr\left[ \mathbf{P}_{\{j\}} \cdot \left(\lambda_0\ketbra{0}{0} \otimes \Lambda^{\otimes j-1} \right)\right] \cdot \tr\left[ \Pi_{\{j\}^c} \cdot \Lambda ^{\otimes n-j} \right]\\
           &=\frac{1}{n p_n} \sum_{j = 1}^n \lambda_0\tr\left[ \ketbra{0}{0} \cdot \Lambda^{ j-1} \right] p_{n-j}\\
           &=\frac{\sum_{j = 1}^np_{n-j} \lambda_0^j}{n p_n},
       \end{aligned}
   \end{equation}
which completes this proof.
\end{proof}

% \begin{shaded}
\begin{lemma}
   Let $\delta$ be the noise parameter of the depolarizing channel $\cN_{\delta}$, and $\Lambda:=\operatorname{Diag}(\lambda_0\,\cdots,\lambda_{d-1})$ be the eigenvalue matrix of a $d$-dimensional noisy state, where  $\lambda_0:=1-\frac{d-1}{d}\delta$, $\lambda_j=\frac{\delta}{d}$, for $j=1,\cdots,d-1$. Denote $Q_2:=\int\cN_\delta(\psi)^{\otimes 2}\otimes\psi_3\,d\psi$, $R_2:=\int\cN_\delta(\psi)^{\otimes 2}\otimes I_3\,d\psi$, we have 
    \begin{equation}
        f_2R_2-Q_2^{T_3}\geq 0,
    \end{equation}
    where $f_2:=\frac{1}{2p_2}(\lambda_0+\lambda_0^2)$, $p_2 := \frac{1}{2}(1+\tr[\Lambda^2])$, and $T_3$ denotes the partial transpose on the third system.
    \label{lem:case_of_two}
\end{lemma}
% \end{shaded}
\begin{proof}
    Denote $\omega:=\frac{1}{t_2}\left(f_2R_2-Q_2^{T_3}\right)$ with $t_2:=\tr\left[f_2R_2-Q_2^{T_3}\right] $, which is a hermitian operator. Observing that $t_2 = f_2 d - 1 > 0$, which means that we need to focus on the positivity of $\omega$. Specifically, denote $\mathbf{P}_3(c)$ as the permutation operator of the symmetric group $\cS_3$, for $c\in\cS_3$, and $D(n,d):=\tbinom{n+d-1}{n}$. One can find that the operator $\omega$ is a linear combination of six operators $\mathbf{P}_3(c)^{T_3}$, i.e.,
    \begin{equation}
        \begin{aligned}
        \omega=&\frac{f_2}{t_2} \int\cN_\delta(\psi)^{\otimes 2}\otimes I_3\,d\psi - \frac{1}{t_2}\int\cN_\delta(\psi)^{\otimes 2}\otimes\psi^{T_3}\,d\psi \\
        =& \frac{f_2}{t_2} \int\left((1-\delta)\psi + \frac{\delta}{d}I \right)^{\otimes 2}\otimes I_3\,d\psi - \frac{1}{t_2}\int\left((1-\delta)\psi + \frac{\delta}{d}I \right)^{\otimes 2}\otimes\psi^{T_3}\,d\psi\\
        =& \frac{f_2}{t_2} \left((1-\delta)^2\frac{\Pi_{12}}{D(2,d)} + \frac{2\delta-\delta^2}{d^2} I_{12} \right)\otimes I_3 \\
        &\quad-\frac{1}{t_2}\left((1-\delta)^3\frac{\Pi_{123}^{T_3}}{D(3,d)} + \frac{2(1-\delta)^2\delta}{d }\frac{\Pi_{12}}{D(2,d)} \ox I_3 + \frac{3\delta - 2\delta^2}{d^2}I_{123} \right)\\
        =&\frac{1}{t_2}\left(\frac{f_2d(2\delta-\delta^2)-\delta^2}{d^3}-\frac{(1-\delta)\delta}{dD(2,d)}+\frac{f_2(1-\delta)^2}{2D(2,d)}-\frac{(1-\delta)^2}{6D(3,d)}\right)I_{123}\\
        &\quad+\frac{1}{t_2}\left(\frac{f_2(1-\delta)^2}{2D(2,d)}-\frac{(1-\delta)^2}{6D(3,d)}\right)\mathbf{P}_3((12))
        - \frac{1}{t_2}\left(\frac{(1-\delta)^2}{6D(3,d)}+\frac{(1-\delta)\delta}{2dD(2,d)}\right)\mathbf{P}_3((13))^{T_3}\\
        &\quad - \frac{1}{t_2}\left(\frac{(1-\delta)^2}{6D(3,d)}+\frac{(1-\delta)\delta}{2dD(2,d)}\right)\mathbf{P}_3((23))^{T_3}
        -\frac{(1-\delta)^2}{6D(3,d)t_2}\mathbf{P}_3((123))^{T_3}\\
        &\quad -\frac{(1-\delta)^2}{6D(3,d)t_2}\mathbf{P}_3((132))^{T_3},
        \end{aligned}
    \end{equation}
    where $\Pi_{123}$ and $\Pi_{12}$ are the projector acting on $\{1,2,3\}$ and  $\{1,2\}$ systems, respectively. Notice that the operator $\omega$ can be represented by a basis $\{S_+,S_-,S_0,S_1,S_2,S_3\}$~\cite{eggeling2001separability}. This basis can be generated by $X$ and $V$, which is as follows:
    \begin{equation}
        \begin{aligned}
        S_+:&=\frac{I+V}{2}\left(I-\frac{2X}{d+1}\right)\frac{I+V}{2},\, S_-:=\frac{I-V}{2}\left(I-\frac{2X}{d-1}\right)\frac{I-V}{2},\\
        S_0:=\frac{1}{d^2-1}&[d(X+VXV)-(XV+VX)],\,  S_1:=\frac{1}{d^2-1}[d(XV+VX)-(X+VXV)],\\
        &S_2:=\frac{1}{\sqrt{d^2-1}}(X-VXV),\,S_3:=\frac{i}{\sqrt{d^2-1}}(XV-VX),
        \end{aligned}
    \end{equation}
    where $I:=\mathbf{P}_3((1))$, $X:=\mathbf{P}_3((2,3))^{T_3}$ and $V:=\mathbf{P}_3((1,2))$. It satisfies that $X^2=dX$, $V^2=I$, and $XVX=X$. Furthermore, one can find that $\tr[X]=d^2$, $\tr[XV]=d$. Based on these properties, we  have the following equations:
    \begin{equation}
        \begin{aligned}
            s_+:=&\tr[\omega S_+]=\frac{d(2-2\delta+\delta^2)+(2-\delta)\delta}{2d},\\
            s_-:=&\tr[\omega S_-]= - \frac{(d-2)\delta(d(2-\delta)+\delta)^2}{2d(d-1)(d(\delta-2)-2\delta)},\\
            s_0:=&\tr[\omega S_0]=\frac{2\delta(d\delta-d-\delta)}{d(d-1)(d(\delta-2)-2\delta)},\\
            s_1:=&\tr[\omega S_1]=\frac{2\delta(d+\delta-d\delta)}{d(d-1)(d(\delta-2)-2\delta)},\,
            s_j:=\tr[\omega S_j]=0,\,j=2,3.
        \end{aligned}
    \end{equation}
    It is straightforward to check that these equations satisfy the following situations:  $s_+,s_-,s_0\geq 0$, $s_+ + s_- + s_0=1$ and $s_1^2+s_2^2+s_3^2\leq s_0^2$. Notice that these six operators $S_+$, $S_-$, $S_0$, $S_1$, $S_2$, $S_3$ are characterized by the following commutation relations: $S_jS_{\pm}=S_{\pm}S_j=0$, $S^2_{j}=S_0$, for $j=0,1,2,3$ and $S_1S_2=iS_3$ with cyclic permutations. According to Lemma 2 proposed in~\cite{eggeling2001separability}, we obtain $\omega\geq 0$ which means that $f_2R_2-Q_2^{T_3} \geq 0$. We complete this proof.
\end{proof}

\end{document}